\documentclass[acmsmall,screen,nonacm]{acmart}
\usepackage{commath} 
\usepackage{graphicx}
\usepackage{makecell}
\usepackage{stmaryrd}
\usepackage{xcolor}
\usepackage{varwidth}
\usepackage{multirow}
\usepackage[multiple]{footmisc}
\usepackage{comment}
\usepackage{outlines}
\usepackage{csquotes}
\usepackage{listings, listings-rust}
\usepackage{algorithm}
\usepackage{algpseudocode}
\usepackage{pgfplots}
\usepackage{array}
\usepackage{xurl}
\usepackage{cleveref}
\usepackage{subcaption}
\usepackage{booktabs}
\usepackage{geometry}
\usepackage{longtable}
\usepackage{ltablex}
\usepackage{tikz-cd}
\usepackage{algorithmicx}

\usepackage{relsize}  

\usepackage{silence}
\pgfplotsset{compat=1.18}
\usetikzlibrary{shapes,arrows,decorations.pathmorphing,positioning,fit,calc,backgrounds}
\acmConference[unpublished]{manuscript under preparation}{July 2025}{}
\setcopyright{none}

\title[Fully Symbolic Analysis of Loop Locality]{Fully Symbolic Analysis of Loop Locality}
\subtitle{Using Imaginary Reuse to Infer Real Performance}

\author{Yifan Zhu}
\orcid{0009-0000-6718-7787}
\affiliation{
  \institution{University of Rochester}
  \city{Rochester}
  \state{New York}
  \country{USA}
}

\author{Yekai Pan}
\orcid{0009-0000-6718-7787}
\affiliation{
  \institution{University of Rochester}
  \city{Rochester}
  \state{New York}
  \country{USA}
}

\author{Chen Ding}
\orcid{0000-0003-4968-6659}
\affiliation{
  \institution{University of Rochester}
  \city{Rochester}
  \state{New York}
  \country{USA}
}

\author{Yanghui Wu}
\affiliation{
  \institution{University of Rochester}
  \city{Rochester}
  \state{New York}
  \country{USA}
}

\begin{document}

\begin{abstract}
    This paper presents a new theory of locality and its compiler support.  The theory is fully symbolic and derives locality as polynomials, and the compiler analysis supports affine loop nests.  They derive cache-performance scaling in quadratic and reciprocal expressions and are more general and precise than empirical scaling rules.
    
    Evaluated on a benchmark suite of 41 scientific kernels and tensor operations, the compiler requires an average of 41 seconds to derive the locality polynomials. After derivation, predicting the cache miss count for any given input size and cache configuration takes less than a millisecond. 
    Across all tests--with and without loop fusion--the accuracy in the data movement prediction is 99.6\%, compared to simulated set-associative L1 data cache.  
    

\end{abstract}

\maketitle

\lstset{
  language=C,
  basicstyle=\ttfamily\small,
  keywordstyle=\color{blue}\bfseries,
  commentstyle=\color{gray},
  numberstyle=\tiny\color{gray},
  numbersep=5pt,
  frame=single,
  tabsize=2,
  showstringspaces=false
}

\def \fp {\textit{fp}}
\def \ri {\textit{ri}}
\def \mr {\textit{mr}}
\def \hr {\textit{hr}}
\def \rd {\textit{rd}}
\def \rtfp {\textit{rtfp}}

\section{Introduction}\label{sect:intro}

Locality, in Denning's sense, is the tendency for memory references to cluster within small regions of the address space over time \cite{Denning:CSUR21}. It is a fundamental property of program behavior and a determinant of performance for data-intensive applications.  Locality optimization is important in modern compilers~\citep{AllenK:Book01,Aho+:Book06,Wolfe:Book96,Acharya+:PLDI18}, transformation frameworks~\citep{Yi:SPE12,PatabandiH:CC23}, domain-specific language and compiler for image processing~\citep{Ragan-Kelley+:PLDI13} and tensor operations~\citep{Kjolstad+:OOPSLA17,Vasilache+:TACO20,Li+:SC19}.  A complete characterization should quantify the locality by both machine and program parameters.  



This paper presents a fully symbolic technique to analyze loop locality.  Given symbolic loop bounds, our method derives two types of polynomials: the first computes the cache size, and the second the miss count or the miss ratio.  We call these formulas \emph{algebraic locality} because these are algebraic expressions that capture how locality changes with various parameters.  For example, in a single traversal of an \( n \times n \) matrix, the miss-ratio polynomial is \( \frac{1}{b} \), where $b$ is the cache block size. 



The new analysis is based on a measure called the reuse interval (RI). A reuse interval (RI) is defined as the time between two consecutive accesses to the same memory location.  We solve mainly two problems.  The first is a theory on using RI polynomials to derive cache miss polynomials and its correctness.  The second is compiler analysis of RIs in affine loops.  The formal derivation is iterative and takes time linear to the number of symbolic RI values.  The compiler analysis has two passes and uses integer set programming.  

What is new and necessary to algebraic theory is the introduction of imaginary reuses. Traditionally, the RI of a first-touch access is treated as infinite.  However, an infinite RI poses a problem for symbolic analysis. Imaginary reuse is introduced to assign a finite RI value to first-touch accesses.

Previous techniques are symbolic but not algebraic.  Affine loops can be symbolically analyzed by solving integer-set equations~\cite{Kelly+:LCPC99,BeylsD:JSA05,GMM:TOPLAS99,Bao+:POPL18,Falcon:PLDI24}.  Unlike polynomials, these integer equations are linear and cannot have quadratic or reciprocal terms.  As a result, they cannot be fully symbolic.  For example, they model spatial locality for only constant, not symbolic, cache block size.  In addition, their solution time is exponential in the worst case.  In comparison, the derivation has a linear-time cost in the new theory.

The main contributions of the paper are as follows:

\begin{itemize}
    \item \textbf{Algebraic locality theory}, which uses imaginary reuses and an iterative method to derive cache polynomials in linear time.  The theory proves two formal properties: Working-set Correctness and RI Sum Invariance.
    
    \item \textbf{Affine-loop compiler analysis}, which translates the affine dialect of MLIR into parametric polytopes, uses a two-pass algorithm to count all possible lengths of reuse intervals and their counts as polynomials.
    \item \textbf{Implementation and evaluation}, including 30 scientific kernels~\cite{polybench} and 11 common tensor operations~\cite{blacher2024einsum}, before and after loop fusion optimization, for both analysis speed and accuracy compared against cache simulation and hardware counters.  
\end{itemize}

One use of cache polynomials is scaling analysis.  For example, the $\sqrt{2}$ rule is well known for database systems (scaling buffer pool sizes), file system caches, and Web caching~\citep{Hartstein+:JILP08}.  Cache polynomials show cache scaling for affine loops and have a greater precision than scaling rules.  \Cref{sec:min-max} shows two cache scalings both follow the $\sqrt{2}$ rule, but their miss ratio differ by a factor of two. 



The locality property determines data movement, i.e., the miss count, but does not fully determine the running time, which significantly depends on latency hiding techniques like out-of-order execution and data prefetching.  Latency tolerance, however, does not decrease the amount of data movement.
In addition, we consider sequential programs only and focus on the miss ratio of a single-level cache, where we consider set associativity and the actual hardware cache.  While locality analysis enables program optimizations, this paper focuses on analysis, not optimization, except that we evaluate the analysis with and without optimization, in particular, loop fusion.


\section{The Theory of Algebraic Locality}

This section describes the theory and its new concept called imaginary reuses.  The flow chart in \Cref{fig:img_overview} shows the structure of the subsections: infinite repeat (\Cref{sec:inf-repeat}) produces imaginary reuses which are used in deriving locality polynomials (\Cref{sec:denning,sec:ws-correctness}) and symbolic testing (\Cref{sec:0-check}).  \Cref{sec:LRU_approx} describes LRU approximation. Finally, \Cref{sec:examples}

\begin{figure}[h!]
    \centering
    \resizebox{0.618\linewidth}{!}{
    \begin{tikzpicture}[
        node distance=0.75cm and 2cm,
        flowbox/.style={draw, rounded corners, minimum width=2.5cm, minimum height=0.8cm, align=center, font=\small, fill=blue!10},
        arrow/.style={-stealth, thick}
    ]
        \node[flowbox] (infinite) {Infinite Repeat};
        \node[flowbox, below=of infinite] (imaginary) {Imaginary Reuses};
        \node[flowbox, below left=0.75cm and 1.5cm of imaginary] (denning) {Cache Size and \\Miss Polynomials};
        \node[flowbox, below right=0.75cm and 1.5cm of imaginary] (testing) {Symbolic\\Testing};
        
        \draw[arrow] (infinite) -- (imaginary);
        \draw[arrow] (imaginary) to[out=-150, in=90] node[left, font=\scriptsize, align=left, pos=0.6] {Denning\\Recursion\citep{DenningS:CACM72,Yuan+:TACO19}} (denning);
        \draw[arrow] (imaginary) to[out=-30, in=90] node[right, font=\scriptsize, align=left, pos=0.6] {RI Sum\\Invariance} (testing);
    \end{tikzpicture}
    }
    \caption{Infinite repeat creates imaginary reuses which are used in symbolic locality analysis}
    \label{fig:img_overview}
    \Description{A flowchart with four boxes. The top box is “Infinite Repeat,” pointing downward to “Imaginary Reuses.” From “Imaginary Reuses,” two arrows branch: one to the lower-left box labeled “Cache Size and Miss Polynomials” via Denning Recursion, and one to the lower-right box labeled “Symbolic Testing” via RI sum invariance. The figure illustrates how imaginary reuses created by infinite repetition feed both symbolic cache polynomial derivation and symbolic testing.}
\end{figure}
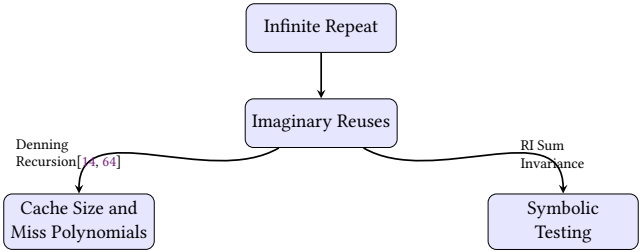

\subsection{Background: Reuse Intervals and Denning Recursion}
\label{sec:denning}

In the view of the memory system, a program execution is a sequence of memory accesses.  Each consecutive pair of accesses to the same data item is a data reuse. The \emph{reuse interval (RI)} is the number of other accesses between the reuse itself. In this paper, we use the logical time, where each access increases the time by one.  The RI value between two accesses is the difference in their logical time.  

Let $P(\ri)$ be a RI distribution, where $P(\ri=x)$ is the portion of RIs whose value is equal to $x$. 
Denning Recursion~\citep{DenningS:CACM72} is a pair of equations to calculate the time-window miss ratio and the average working-set size from the RI distribution as follows:
\begin{align*}
m(x) & = P(\ri > x) \\
s(x+1) & = s(x) + m(x)
\end{align*}
where $m(x)$ is the miss ratio, and $s(x)$ the working-set size. The two equations are defined for all integers $x \ge 0$, with an initial value of $s(0)=0$ and $m(0)=P(\ri>0)=1$.  To compute all results $m,s$, Denning Recursion iterates all RI values $x$ in increasing order.  At each step, the miss ratio is the portion of reuses with RI greater than $x$, and the working-set size increases by the miss ratio.  

The time-window miss ratio $m(x)$ is a function of time.  The working-set size is an average value, which is usually fractional.  As described in \Cref{sec:ws-correctness}, the new theory will use these to approximate the LRU cache size and miss ratio.

\subsection{Infinite Repeat and Imaginary Reuses}
\label{sec:inf-repeat}

In a program, data access may be using an item for the first time, which is a \emph{first-touch access}; otherwise, it is a data reuse.\footnote{The attribute of a reuse can be associated with either the start or the end of the reuse window.  If it is the former, it is a forward use; otherwise, it is a backward reuse~\citet{BeylsD:JSA05}.  In this paper, all RIs are backward RIs.}

A first-touch access is a cold-start miss in any cache~\citep{Hill:Dissertation}. For Denning Recursion, however, first-touch accesses present a fundamental dilemma. If we 
assign infinite reuse intervals to first-touch accesses, then $P(\ri = \infty)>0$ causes the working-set size to diverge to infinity ($s(\infty)=\infty$), making the analysis meaningless. On the other hand, if we exclude first-touch accesses from the analysis entirely, the working-set size remains bounded but we fail to account for cold-start misses. We call this the cold-start miss dilemma.

The algebraic theory resolves the dilemma by introducing \emph{Infinite Repeat}, where a program repeats an infinite number of times.  A first-touch access is still so in the first run, but it is a reuse in the second run and all other repeats.  We call these reuses \emph{Imaginary Reuses}.  The RI of an imaginary reuse is an \emph{Imaginary RI}.  Each data item’s imaginary reuse interval spans from its last access in one run to its next access in the following run.  
If the item is not reused within a run, this interval equals the program length~$n$.  
Otherwise, it is the gap between the last access in the previous run and the first access in the current run, which can be at most~$n - 1$.  
Hence, the largest possible imaginary reuse interval is~$n$.

Figure~\ref{fig:accesses_reuses} illustrates the technique.  Let a simple program be a series of five accesses.  The figure shows the first run and two repeats of the program: the second run and the $r$th run.  The first, second and fourth of the five accesses are first-touch accesses in the first run, colored red.  After the first run, they become imaginary reuses, colored green.  The arrows connect the source and the sink of imaginary reuses between the first two program runs. The largest imaginary RI value is 5 for the access to `B', which is its sole access in a run.

To distinguish, we refer to data reuses before Infinite Repeat as \emph{real reuses}.  Their values and portions are not changed by Infinite Repeat.  
In \Cref{fig:accesses_reuses}, they are third and fifth accesses in each run (colored blue).

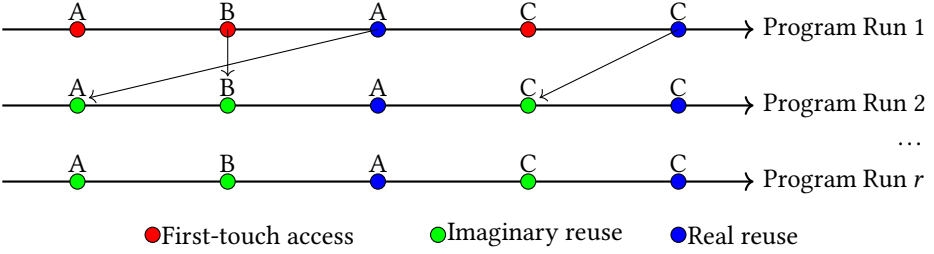
\begin{figure}[htbp]
    \centering
    \begin{tikzpicture}[x=1cm, y=1cm]

        \draw[thick, ->] (0, 3) -- (10, 3) node[right] {Program Run 1};
        \draw[thick, ->] (0, 2) -- (10, 2) node[right] {Program Run 2};
        \draw[thick, ->] (0, 1) -- (10, 1) node[right] {Program Run $r$};

        \draw[fill=red] (1, 3) circle [radius=0.1] node[anchor=south] {A};
        \draw[fill=red] (3, 3) circle [radius=0.1] node[anchor=south] {B};
        \draw[fill=blue] (5, 3) circle [radius=0.1] node[anchor=south] {A};
        \draw[fill=red] (7, 3) circle [radius=0.1] node[anchor=south] {C};
        \draw[fill=blue] (9, 3) circle [radius=0.1] node[anchor=south] {C};

        \draw[fill=green] (1, 2) circle [radius=0.1] node[anchor=south] {A};
        \draw[fill=green] (3, 2) circle [radius=0.1] node[anchor=south] {B};
        \draw[fill=blue] (5, 2) circle [radius=0.1] node[anchor=south] {A};
        \draw[fill=green] (7, 2) circle [radius=0.1] node[anchor=south] {C};
        \draw[fill=blue] (9, 2) circle [radius=0.1] node[anchor=south] {C};

        \node at (12.1,1.5) {\dots};

        \draw[fill=green] (1, 1) circle [radius=0.1] node[anchor=south] {A};
        \draw[fill=green] (3, 1) circle [radius=0.1] node[anchor=south] {B};
        \draw[fill=blue] (5, 1) circle [radius=0.1] node[anchor=south] {A};
        \draw[fill=green] (7, 1) circle [radius=0.1] node[anchor=south] {C};
        \draw[fill=blue] (9, 1) circle [radius=0.1] node[anchor=south] {C};

        \draw[->] (5, 3) -- (1.15, 2.1) node[midway, left] {};
        \draw[->] (3, 3) -- (3, 2.38) node[midway, right] {};
        \draw[->] (9, 3) -- (7.15, 2.1) node[midway, right] {};

        \begin{scope}[shift={(0, 0.3)}]
            \draw[fill=red] (2, 0) circle [radius=0.1] node[right] {First-touch access};
            \draw[fill=green] (5.8, 0) circle [radius=0.1] node[right] {Imaginary reuse};
            \draw[fill=blue] (9, 0) circle [radius=0.1] node[right] {Real reuse};
        \end{scope}

    \end{tikzpicture}
    \caption{An example of Infinite Repeat, with first-touch accesses (red), real reuses (blue), and imaginary reuses (green)}
    \Description{A timeline showing the program runs with first-time accesses and reuses of variables A, B, and C. Red nodes indicate first-time access, green nodes indicate reuse, and blue nodes indicate real reuse.}
    \label{fig:accesses_reuses}
\end{figure}

Imaginary reuses have no direct effect on real reuses.  However, to produce a correct model of cache behavior, Denning Recursion requires that all accesses have a finite RI value.  Imaginary RIs are an “additive.”  Intuitively, their inclusion is necessary as the context of real reuses, enabling the correct computation of whether real reuses result in cache hits or misses. The next two sections first formalize this correctness property and then show how to convert imaginary reuses back to first-time accesses.

\subsection{Working-set Correctness}
\label{sec:ws-correctness}
\label{sec:deriv-rd}

Denning Recursion was derived as a stochastic property, assuming an unending sequence, a stationary process, and a third property that ``guarantees that the time average converges to the stochastic average''~\citep{DenningS:CACM72}.  For deterministic sequences, \citet{Xiang+:ASPLOS13} defined the footprint function, $\fp(x)$, which is the average working-set size of a time window of length $x$.  We use both the working-set theory and the footprint theory to establish the correctness of Denning Recursion when using imaginary reuses.

We prove working-set correctness under Infinite Repeat: $c(x)$ computed in the Denning Recursion is the footprint $\fp(x)$, that is, the average working-set size of a deterministic sequence.  This means that Denning Recursion is correct for any program under Infinite Repeat and does not require any stochastic assumption about such cases.

Given a finite-length sequence, the Xiang formula computes the footprint $\fp(x)$, the average working-set size of all windows of length $x$ \citep{Xiang+:PACT11}.  Let the set $\{r_i\}$ be the values of all RIs and the sets $\{f_d\},\{l_d\}$ the first- and last-access times for all data $d$.  We slightly change the representation of the Xiang formula to refer to all boundary accesses as $\{b_d\}$, which includes both the first and last access.  Each data item $d$ has two $b_d$s.
Let $n$ be the number of accesses (the length of a trace), and $m$ be the number of distinct data (the size of the data).  The Xiang formula computes the footprint as follows

\begin{align*}
    \fp(x) =& m - \frac{\sum_{r_i > x} (ri -x ) - \sum_{b_d > x} (b_d -x)}{n-x+1} 
\end{align*}

\noindent where $x=0, \dots, n$.  
We further define the reuse-term footprint by\footnote{In a sequence of length $n$, the largest possible RI value is $n-1$.  With Infinite Repeat, the largest possible value is $n$, that is, when a symbol appears only once in the sequence (before repeating and becomes an imaginary reuse afterward).  An example is B in \Cref{fig:accesses_reuses}.}

$$\rtfp(x) = m - \sum_{i=x+1}^n (i-x)P(\ri=i)$$

Intuitively, $\rtfp(c)$ counts how many RI values exceed the cache size $c$, weighted by their frequency. Under infinite repeat, this footprint converges to zero because every access eventually becomes a reuse. The following theorem formalizes this property and shows that the working-set size correctly accounts for all distinct data items accessed.

\begin{theorem}[Working-set Correctness]
\label{thm:ws-correctness}
    Under Infinite Repeat, Denning Recursion computes the average working-set size correctly.
\end{theorem}
\begin{proof}
    Let $n$ be the length of a program and $r$ be the number of repeats.  We first prove \emph{$\rtfp(x)$ correctness}.  
    As $r$ becomes infinite, 
    the boundary terms in the Xiang formula have a finite value and can be removed because the denominator, $rn-x+1$, becomes infinite.  Furthermore, each RI has a finite value at most $n$, its contribution is exactly $\frac{1}{n}$.
    Hence, the reuse-term footprint is correct because it converges to the Xiang footprint.  
    Furthermore, Denning Recursion is equivalent to the reuse-term footprint, $s(x) = \fp(x)$.  We prove it inductively.  In the base case, $s(0)=\fp(0)=0$ (see also the following lemma).  Assuming $s(x)=\fp(x)$, it is straightforward to verify $\fp(x+1)-\fp(x)=s(x+1)-s(x)=P(\ri>x)$.  Combining it with $\rtfp(x)$ correctness, we have 
    $$s(x) = \rtfp(x) = \lim_{r \rightarrow \infty} \fp(x) $$
\end{proof}

Under Infinite Repeat, the average working-set size is the same whether it is computed using Denning Recursion or using the footprint.  In other words, imaginary reuses guaranty the correctness of Denning Recursion when used on a finite-length program.  



Omitted for lack of space, we can also prove that when the working-set (cache) size equals or exceeds the data size, the miss ratio drops to $0$.
A reader may question that in real executions, cold-start misses are inherent and not zero. Next, we consider the LRU cache and cold-start misses.

\subsection{RI Sum Invariance and Symbolic Testing}
\label{sec:0-check}

From Working-set Correctness (\Cref{thm:ws-correctness}), we can derive the following lemma by taking the timescale parameter at $x=0$.  When $x=0$, a zero-length window contains no data access, the working-set size or footprint is therefore zero.  We call it the Zero Footprint lemma.

\begin{lemma}[Zero Footprint]
\label{lma:zero-fp}
    Under Infinite Repeat, $\rtfp(0) = m - \sum_{\ri=0}^n \ri\ P(\ri) = 0$, where $m$ is the data size, $P(\ri)$ are the portions of all $n$ distinct RI values $\ri$.  
\end{lemma}

Re-arranging the terms, we have $\sum_{\ri=0}^n \ri\ P(\ri) = m$.  If we view the $n$ RI values and $n$ RI portions as vectors, the left-hand side is their dot product.  We call the equation
\emph{RI Sum Invariance}.  The invariance states that \emph{the dot product of the value vector and the portion vector must equal the data size,} or $\mathbf{v} \cdot \mathbf{p} = m$.  

The invariance provides a \emph{Symbolic RI Test} for checking for errors in symbolic RI values or portions. 
From the dot-product rule: {\bf all product terms must cancel each other}.  For any program under Infinite Repeat, this check must pass; otherwise, there must be an error in one of the RI values or one of their portions. 

The symbolic RI test is not a \emph{proof} of RI correctness.  It checks that a set of symbolic terms must all cancel each other.  It rules out single-term errors, but multiple errors can cancel each other out and still pass the test.  Statistically, the test provides greater assurance of correctness as the number of symbolic RI values increases.  The more complex a program is, the more RI values it has, and the more likely an error may occur.  Though imperfect, this test offers a practical safeguard, whether the analysis is done manually or automated by a compiler.

\subsection{LRU Approximation}
\label{sec:LRU_approx}

Our work approximates the LRU cache size with the \emph{average working-set size $c(x)$}. 

\theoremstyle{definition}
\newtheorem{assump}{Assumption}

\begin{assump}[LRU Approximation]  The miss ratio of the LRU cache can be approximated by the miss ratio of the working-set cache by:
    $$\mr(c(x)) \approx m(s(x))$$
where $c(x)=s(x)$.
\end{assump}

Since $s(x)$ is fractional, the approximated LRU cache size $c(x)$ may also be fractional. To interpret fractional cache sizes, note that $s(x)$ is monotone and increases by at most one at each step ($s(x+1)-s(x) \le 1$ for $x>0$). Therefore, every integer cache size $k$ is either exactly $s(x)$ for some $x$, or it lies between two consecutive values $s(x)$ and $s(x+1)$ that differ by less than one. By monotonicity, the miss ratio for integer cache size $k$ can be bounded by the miss ratios at the surrounding fractional cache sizes. Hence, the approximation effectively computes miss ratios for \emph{consecutive cache sizes}, allowing interpolation to any integer value.


Many past techniques use the LRU approximation, including time-to-RD conversion in program analysis~\citep{Shen+:POPL07}, Featherlight RD using hardware counters~\citep{Wang+:HPCA19}, StatStack~\citep{EklovH:ISPASS10}, Average Eviction Time (AET) for storage caches~\citep{Pan+:TOS21} (as reviewed in \Cref{sec:rel}).  They showed generally high accuracy in large-scale traces.  Unlike these studies, we use the approximation for loop kernels.

\paragraph{Moving-cliff Error} 
Loop kernels were studied by \citet{Chen+:ISMM21}.  They found mainly one type of error and its cause.  If a kernel has \emph{phases}, i.e., multiple loop nests, the prediction is often inaccurate.  The error manifests itself as a ``moving cliff'': when comparing the prediction with the LRU simulation, the predicted drop in miss ratio occurs earlier.  The miss ratio, both before and after the drop, is predicted correctly, but the predicted cache size is an underestimate. 



\paragraph{Cold-start Misses}
To recover the actual miss ratio, we ``undo'' the effect of Infinite Repeat. Denning Recursion produces both cache size and miss ratio. The computed cache size is correct because it accounts for all distinct data items, whether accessed via real reuses or imaginary reuses. The miss ratio, however, treats first-time accesses as imaginary reuses (hits from the prior iteration). In reality, first-time accesses in a single program execution should be cold misses~\citep{Hill:Dissertation}. Therefore, we adjust the miss ratio by converting imaginary reuse hits back to misses. The hit portion attributed to imaginary reuses represents the cold-start miss ratio of the original single-execution program.


\subsection{Example Derivations: Naive Matrix Multiplication}
\label{sec:examples}

In this section, we use naive matrix multiplication as a concrete example to illustrate how symbolic RI distributions are derived and transformed via Denning Recursion into cache-miss polynomials.

Matrix multiplication is a fundamental operation in scientific computing and serves as an accessible example for locality analysis. We consider the standard triply nested loop for multiplying two $n \times n$ matrices $A$, $B$, and $C$. To simplify the analysis, we assume that each cache block stores 8 data elements (block size $b = 8$).

\begin{lstlisting}[basicstyle=\footnotesize\ttfamily]
for (int i = 0; i < n; i++) 
  for (int j = 0; j < n; j++)
    for (int k = 0; k < n; k++)
      C[i][j] += A[i][k] * B[k][j];
\end{lstlisting}

Each innermost loop iteration is modeled as four memory accesses in the order C, A, B, C: the value $C[i][j]$ is first read, then $A[i][k]$ and $B[k][j]$ are accessed, and eventually $C[i][j]$ is written back. This matches the typical memory behavior of a naive implementation, resulting in a total of $4n^3$ accesses. 
We applied the imaginary RI value under infinite repeat to ensure that every access is associated with an RI, as described in \Cref{sec:inf-repeat}, including those corresponding to cold misses or the first access to a data block. After constructing the complete algebraic miss ratio table using all symbolic RI values (including the imaginary RI), we remove the portion corresponding to the imaginary RI when reporting the final miss ratio. This adjustment ensures that the reported miss ratio reflects all the cache misses (including cold misses). 

\begin{table}[!htbp]
    \renewcommand*{\arraystretch}{1.2}
    \caption{Analysis of Naive Matrix Multiplication  (* containing imaginary RIs)}
    \label{tbl:ana-matmul}
    \centering
    \resizebox{0.8\linewidth}{!}{
    \begin{tabular}{c|c|c|c|c|c}
           & \textbf{RI}             & \textbf{P (ri)}                & \textbf{m (ri)}                & \textbf{Cold Miss Ratio} & \textbf{c (ri)}                                               \\ \hline
        0  & $-$                     & $-$                            & $1$                            & $-$                      & $0$                                                           \\
        1  & $1$                     & $\frac{1}{4}$                  & $\frac{3}{4}$                  & $-$                      & $1$                                                           \\
        2  & $3$                     & $\frac{1}{4} - \frac{1}{32n}$  & $\frac{1}{2} + \frac{1}{32n}$  & $-$                      & $3$                                                 \\
        3  & $4$                     & $\frac{7}{32}$                 & $\frac{9}{32} + \frac{1}{32n}$ & $-$                      & $5$                                           \\
        4  & $4n - 28$               & $\frac{1}{32} - \frac{1}{32n}$ & $\frac{1}{4} + \frac{1}{16n}$  & $-$                      & $\frac{9n}{8} - \frac{47}{8} - \frac{31}{32n}$                \\
        5  & $4n$                    & $\frac{7}{32}$                 & $\frac{1}{32} + \frac{1}{16n}$ & $-$                      & $\frac{9n}{8} + \frac{9}{8} + \frac{25}{32n}$                 \\
        6*  & $4n^2 - 28n$            & $\frac{1}{32}$                 & $\frac{3}{32n}$               & $\frac{1}{32n}$          & $\frac{n^2}{8} + \frac{3n}{8} - \frac{7}{8} + \frac{25}{32n}$ \\
        7* & $4n^3 - 4n^2 + 4n - 28$ & $\frac{3}{32n}$                & $\frac{3}{32n}$                & $\frac{2}{32n}$          & $\frac{3n^2}{8} - \frac{n}{8} + \frac{9}{8} - \frac{31}{32n}$ \\
        8* & $4n^3 - 32n + 3$        & $\frac{1}{32n}$                & $\frac{3}{32n}$                & $\frac{3}{32n}$          & $\frac{3n^2}{8}$                                              \\
    \end{tabular}
    }
\end{table}

The RI distribution and the derived cache-miss polynomials are shown in Table~\ref{tbl:ana-matmul}. The table marks the row index of imaginary reuses with a $*$.  The sum of their RI portions is the cold-start miss ratio.


\paragraph{The RI Sum Invariance}
The second and third columns of \Cref{tbl:ana-matmul} list eight RI values and their corresponding portions. By the RI Sum Invariance, their dot product must equal the total data size, or formally, $\mathbf{v} \cdot \mathbf{p} = m$, where both the RI values ($\mathbf{v}$) and portions ($\mathbf{p}$) are polynomials in $n$.   

Computing this dot product yields a polynomial with almost 20 terms, most of which are non-constant. Upon summation, all constant terms sum to zero, all linear terms cancel, and all reciprocal terms likewise vanish. The only surviving terms are quadratic, and they collectively sum to $\frac{3n^2}{8}$ --- exactly the total number of data blocks across the three $n \times n$  matrices (8 elements per block). Thus, the RI Sum Invariance holds, and the test passes. 


\section{Compiler Design and Implementation}
\label{sec:compiler}

Our framework leverages MLIR \cite{MLIR} as a flexible intermediate representation bridge. By using the MLIR Affine dialect, our Symbolic Locality Compiler can be integrated into existing workflows—such as the Polygeist pipeline \cite{polygeistPACT}—without requiring changes to the frontend or backend. This design allows us to improve locality analysis and optimization across various programming languages and tools that utilize the MLIR infrastructure.

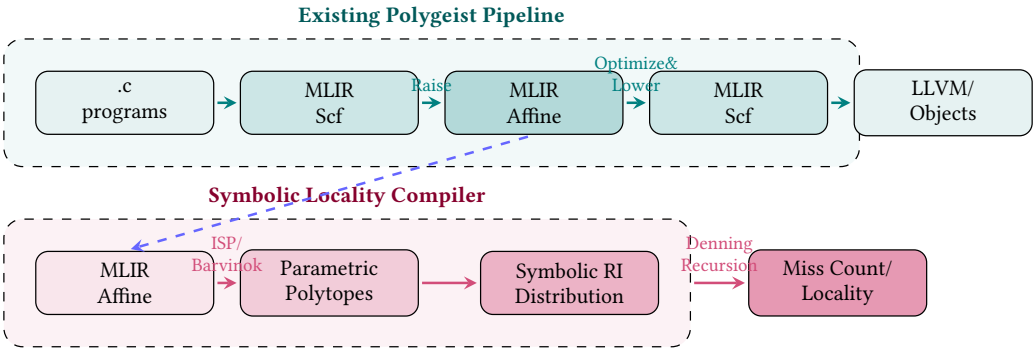
\begin{figure}[htbp]
    \centering
    \resizebox{\linewidth}{!}{
    \begin{tikzpicture}[
        node distance=0.3cm and 0.6cm,
        box/.style={draw, rounded corners, minimum width=2cm, minimum height=0.7cm, align=center, font=\scriptsize, fill=white},
        group/.style={draw, dashed, rounded corners=6pt, inner sep=10pt},
        arrow/.style={-stealth, thick, shorten >=1pt, shorten <=1pt},
        connection/.style={-stealth, thick, dashed, color=blue!60, shorten >=2pt, shorten <=2pt}
    ]
    
    \begin{scope}[local bounding box=polygeist]
        \node[box, fill=teal!10] (c) at (-6.5, 0) {.c\\programs};
        \node[box, fill=teal!20] (scf1) at (-4.2, 0) {MLIR\\Scf};
        \node[box, fill=teal!30] (affine1) at (-1.9, 0) {MLIR\\Affine};
        \node[box, fill=teal!20] (scf2) at (0.4, 0) {MLIR\\Scf};
        \node[box, fill=teal!10] (llvm) at (2.7, 0) {LLVM/\\Objects};
        
        \draw[arrow, color=teal!100] (c) -- (scf1);
        \draw[arrow, color=teal!100] (scf1) -- node[above, font=\tiny, pos=0.5] {Raise} (affine1);
        \draw[arrow, color=teal!100] (affine1) -- node[above, font=\tiny, align=center, pos=0.5] {Optimize\&\\Lower} (scf2);
        \draw[arrow, color=teal!100] (scf2) -- (llvm);
    \end{scope}
    
    \begin{scope}[on background layer]
        \node[group, fill=teal!5, fit={(c) (scf2)}, label={[font=\scriptsize\bfseries, color=teal!70!black]above:Existing Polygeist Pipeline}] (polygeist-box) {};
    \end{scope}
    
    
    \begin{scope}[local bounding box=newpipe, shift={(0, -2)}]
        \node[box, fill=purple!10] (affine2) at (-6.5, 0) {MLIR\\Affine};
        \node[box, fill=purple!20] (polytopes) at (-4.2, 0) {Parametric\\Polytopes};
        \node[box, fill=purple!30] (ridist) at (-1.5, 0) {Symbolic RI\\Distribution};
        
        \draw[arrow, color=purple!70] (affine2) -- node[above, font=\tiny, align=center, pos=0.5] {ISP/\\Barvinok} (polytopes);
        \draw[arrow, color=purple!70] (polytopes) -- (ridist);
    \end{scope}
    
    \begin{scope}[on background layer]
        \node[group, fill=purple!5, fit={(affine2) (ridist)}, label={[font=\scriptsize\bfseries, color=purple!70!black]above:Symbolic Locality Compiler}] (newpipe-box) {};
    \end{scope}
    
    \node[box, fill=purple!40, right=1cm of ridist] (metrics) {Miss Count/\\Locality};
    \draw[arrow, color=purple!70] (newpipe-box.east) -- node[above, font=\tiny, align=center, pos=0.5] {Denning\\Recursion} (metrics);
    
    \draw[connection] (affine1.south) -- node[right, font=\tiny, align=center, xshift=2pt, pos=0.5] {} (affine2.north);
    
    \end{tikzpicture}
    }
    \caption{Workflow diagram showing the existing Polygeist pipeline and our Symbolic Locality Compiler, interconnected via MLIR Affine representation.}
    \Description{A workflow diagram with two isolated groups in horizontal layout: the existing Polygeist pipeline showing .c programs through MLIR Scf, raise to MLIR Affine, optimize and lower back to MLIR Scf, then to LLVM/Objects; and the Symbolic Locality Compiler showing MLIR Affine transforming through Parametric Polytopes to Symbolic RI Distribution and outputting Miss Count/Locality Metrics via Denning Recursion.}
    \label{fig:workflow-diagram}
\end{figure}

\subsection{Affine Program Definition}

We define the following program structure as our analysis target, simplified from MLIR's \texttt{affine} dialect. Our framework does not permit symbols as coefficients within affine expressions. \Cref{tbl:affine-syntax} presents the syntax in two parts: the left column defines core program constructs (loops, conditionals, blocks, and memory accesses), while the right column defines auxiliary constructs for affine expressions and constraint sets. Here $\mathcal S$ is the symbol set, \textit{array} identifies disjoint memory objects, \textit{constant} denotes integer literals, and \textit{ivar} $= i_{level}$ denotes loop index variables at a given nesting level.

\begin{table}[htbp]
\small
\centering
\caption{Affine program syntax}
\setlength{\tabcolsep}{8pt}
\begin{tabular}{@{}ll@{}}
\toprule
\textbf{Core Syntax} & \textbf{Auxiliary Definitions} \\
\midrule
$\begin{array}{r@{~}c@{~}l}
    \textit{prog} &=& (\mathcal S, \textit{stmt}) \\
    \textit{stmt} &=& \textit{loop} \mid \textit{block} \mid \textit{if} \mid \textit{access} \\
    \textit{loop} &=& \textbf{for} ~ \textit{expr} ~ \textbf{to} ~ \textit{expr} ~ \textbf{step} ~ \textit{constant} \\
    && \textbf{do} ~ \textit{stmt} ~ \textbf{endfor} \\
    \textit{block} &=& \textit{stmt} ~ (; \textit{stmt})+ \\
    \textit{if} &=& \textbf{if} ~ \textit{set} ~ \textbf{then} ~ \textit{stmt} \\
    && \textbf{else}  ~ \textit{stmt} ~ \textbf{endif} \\
    \textit{access} &=& \textbf{access} ~ \textit{array} [\textit{expr} (, \textit{expr})*]
\end{array}$ &
$\begin{array}{r@{~}c@{~}l}
    \textit{atom} &=& \textit{constant} \mid \textit{ivar} \mid s \in \mathcal{S} \\
    \textit{expr} &=& \textit{atom}  \mid \textit{expr} \pm \textit{expr} \\
    && \mid \textit{constant} \times \textit{expr} \\
    && \mid \left\lceil \frac {\textit{expr}}{ \textit{constant}} \right\rceil \mid \left\lfloor \frac {\textit{expr}}{ \textit{constant}} \right\rfloor \\[1ex]
    \textit{constraint} &=& \textit{expr} ~ \mathbf{= 0} \mid \textit{expr} ~ \mathbf{< 0} \\
    \textit{conj} &=& \textit{constraint} ~ (\wedge ~ \textit{constraint})^* \\
    \textit{set} &=&  \textit{conj} ~ (\vee ~ \textit{conj})^*
\end{array}$ \\
\bottomrule
\end{tabular}
\Description{A two-column table showing affine program syntax: the left column defines core constructs (prog, stmt, loop, block, if, access) and the right column defines auxiliary constructs (atom, expr, constraint, conj, set) for affine expressions and constraint sets.}
\label{tbl:affine-syntax}
\end{table}

\subsection{Affine Representation of Timestamp and Memory Access}\label{sect:timestamp-space}

To analyze reuse intervals, we need to obtain the set of access timestamps. This differs from the iteration space, as timestamps are associated with each memory access and do not necessarily describe the values of loop index variables. We demonstrate a translation that converts an affine program into an integer vector space, summarizing the formal semantics of the algorithm in \Cref{fig:timestamp-construction} and providing a concrete example in \Cref{fig:algo-example}.

We define an operation $\llbracket \cdot \rrbracket(\mathcal M; \ell)$ to convert the program representation $p \in \mathcal P$ into the desired space $\mathcal T \subseteq \mathbb Z^n$ for some $n$, under context $\mathcal{M}$ and $\ell$. The following concepts are used in the translation: 

\begin{itemize}
    \item $\mathcal M$ maps index variables (by loop level) to vector components, where $|\mathcal{M}|$ equals the number of in-scope index variables. It normalizes strided loop indices into continuous iteration space.
    \item $\ell$ counts the timestamp vector dimensions used up to the current construct, incremented when entering loops or statement blocks to distinguish memory accesses.
    \item $\mathcal L(s)$ denotes the vector dimension of space $s$.
    \item $x |_{\mathcal M}$ denotes substituting index variables in expression $x$ (or constraint expressions when $x$ is a set) with $\mathcal M$.
    \item $a \cup^\uparrow b$ and $a \cap^\uparrow b$ denote standard set union and intersection when $\mathcal L(a) = \mathcal L(b)$; otherwise, the lower-dimensional space is lifted by adding unconstrained trailing dimensions. We may use corresponding $\sqcup, \sqcap$ variants to illustrate operations on disjoint subsets.
\end{itemize}

\Cref{fig:timestamp-construction} defines the translation algorithm inductively. The base cases involve memory accesses, where the algorithm returns an unconstrained set. In other cases, the algorithm applies recursively onto substatements, aligning space dimensions and adding constraints imposed at parent to refine the space.

\begin{figure}[htbp]
\small
\setlength{\abovedisplayskip}{4pt}
\setlength{\belowdisplayskip}{4pt}
\setlength{\jot}{3pt}
\begin{align*}
    \llbracket \cdot \rrbracket(\mathcal M; \ell) &: \mathcal P \to \bigcup_{n \in \mathbb N}\textbf{Pow}(\mathbb Z^n) \\
    \llbracket (\mathcal S, \textit{s}) \rrbracket(\mathcal M; \ell) &= \llbracket \textit{s} \rrbracket(\mathcal M; \ell) \\
    \llbracket  \textbf{access} ~ A ~ [e_0, ..., e_n] \rrbracket(\mathcal M; \ell) &= \mathbb{Z}^\ell \\[2pt]
    \llbracket  \textbf{for} ~ \textit{l} ~ \textbf{to} ~ \textit{u} ~ \textbf{step} ~ \textit{c} ~ \textbf{do} ~ \textit{b} ~ \textbf{endfor} \rrbracket(\mathcal M; \ell) = 
    \textbf{let} ~ \mathcal{M}' ~  &= ~ \mathcal M \cup \{ |\mathcal M|  \mapsto (l \mid_{\mathcal M} + c \times v[\ell]) \} ;
    s~ = \llbracket \textit{b} \rrbracket(\mathcal M'; \ell + 1) \\[-6pt]
    &~\textbf{in} ~ \left. s \cap^\uparrow \middle\{v \in \mathbb{Z}^{\mathcal L(s)} : \left(v[\ell] \ge 0\right) \wedge \middle( v[\ell] < \middle\lceil \frac{u \mid_{\mathcal M} - l \mid_{\mathcal M}}{c}\middle\rceil\middle) \middle\}\right. \\[2pt]
    \llbracket  s_0; s_1; ...; s_n\rrbracket(\mathcal M; \ell) =
    \textbf{let} ~\displaystyle\mathop{\forall}_{i = 0}^n s'_i &= \llbracket s_i\rrbracket(\mathcal M; \ell + 1) ~ ;
    \mathcal{L}^\ast = \max_{i = 0...n} \mathcal{L}(s_i) \\[-6pt]
    &~\textbf{in}~ {\bigsqcup_{i = 0}^n}^\uparrow \left( s_i' \cap^\uparrow \middle\{ v \in \mathbb{Z}^{\mathcal L^\ast} : (v[\ell] = i) \wedge \middle( \bigwedge_{j = \mathcal{L}(s'_i)}^{\mathcal L^\ast - 1} v[j] = 0 \middle) \middle\}\right) \\[2pt]
    \llbracket  \textbf{if} ~ c ~ \textbf{then} ~ t ~ \textbf{else}  ~ e ~ \textbf{endif} \rrbracket(\mathcal M; \ell) =
    \textbf{let} ~ t' ~ &= \llbracket t \rrbracket(\mathcal M; \ell) ~ ; ~ e' ~ = \llbracket e  \rrbracket(\mathcal M; \ell) ~;~\mathcal{L}^\ast = \max \{ \mathcal{L}(t'), \mathcal{L}(e') \} \\[-6pt]
    &~\textbf{in}~ \left( t' \cap^\uparrow \middle\{ v \in \mathbb{Z}^{\mathcal L^\ast} : c|_\mathcal{M} \wedge \middle( \bigwedge_{j = \mathcal{L}(t')}^{\mathcal L^\ast - 1} v[j] = 0 \middle) \middle\}\right) \\[-6pt]
    &~\sqcup \left( e' \cap^\uparrow \middle\{ v \in \mathbb{Z}^{\mathcal L^\ast} : \neg c|_\mathcal{M} \wedge \middle( \bigwedge_{j = \mathcal{L}(e')}^{\mathcal L^\ast - 1} v[j] = 0 \middle) \middle\}\right)
\end{align*}
\caption{Formal Semantics of Timestamp Space Construction}
\Description{Mathematical definition of timestamp space construction semantics, showing formal notation for loops, blocks, conditionals, and memory access statements with their corresponding timestamp space representations.}
\label{fig:timestamp-construction}
\end{figure}

The timestamp space is ordered lexicographically. An order-preserving bijection can be formed between the memory access trace of the original program $p$ and the timestamp space $\mathcal T = \llbracket p \rrbracket(\emptyset, 0)$. The size of the timestamp space matches the length of the program's access trace.

\begin{example}
\label{exp:local-example}

Consider the program shown on the left in \Cref{fig:algo-example}. We use $s_i$ to denote the single statement at line $i$ and $b_i$ to denote the entire block starting at line $i$. The table on the right summarizes the local input context and output for each $s_i$ and $b_i$. Notably, $s_2$, $s_5$, $s_8$, and $s_9$ are leaf statements. The top-level result $s_1$ is derived step by step from these inner results. The second mapping $1 \mapsto d_0 + 4d_2 - 512$ in the context $\mathcal{M}$ at $s_5$ is not used here but will become important when constructing the access map in later sections.

\begin{figure}[htbp]
\small
\centering
\begin{minipage}[t]{0.48\textwidth}
\vspace{0pt}
\begin{algorithmic}[1]
\For{$0$ \textbf{to} $2026$ \textbf{step} $1$}
    \State \textbf{access} $A[i_0]$;
    \If{$i_0 > 1024$}
        \For{$i_0$ - 512 \textbf{to} $i_0 + 512$ \textbf{step} $4$}
            \State \textbf{access} $B[i_0, i_1]$
        \EndFor
    \Else
        \State \textbf{access} $A[i_0 + 512]$;
        \State \textbf{access} $C[2 i_0 + 1]$
    \EndIf
\EndFor
\end{algorithmic}
\label{algo:example}
\vspace*{\fill}
\end{minipage}%
\hfill
\begin{minipage}[t]{0.48\textwidth}
\vspace{0pt}
\label{tbl:tspace-derivation}
\scriptsize
\setlength{\tabcolsep}{3pt}
\setlength{\arrayrulewidth}{0.4pt}
\begin{tabular}{c|c|l}
\toprule
Stmt & $\ell$ & Result \\
\hline
$s_1$ & 0 & $d_0 \in [0, 2026), d_1 = 0, d_2 = 0$ \\
 & & $\vee~ d_0 \in [0, 1024), d_1 = 1, d_2 \in [0, 256)$ \\
 & & $\vee~ d_0 \in [1024, 2026), d_1 = 1, d_2 \in [0, 2)$ \\
\hline
$b_2$ & 1 & $d_1 = 0, d_2 = 0$ \\
 & & $\vee~ d_0 < 1024, d_1 = 1, d_2 \in [0, 256)$ \\
 & & $\vee~ d_0 \ge 1024, d_1 = 1, d_2 \in [0, 2)$ \\
\hline
$s_2$ & 2 & $\mathbb{Z}^2$ \\
\hline
$s_3$ & 2 & $d_0 < 1024, d_2 \in [0, 256)$ \\
 & & $\vee~ d_0 \ge 1024, d_2 \in [0, 2)$ \\
\hline
$s_4$ & 2 & $d_2 \in [0, 256)$ \\
\hline
$s_5$ & 3 & $\mathbb{Z}^3$ \\
\hline
$b_8$ & 2 & $d_2 \in [0, 2)$ \\
\hline
$s_8$ & 3 & $\mathbb{Z}^3$ \\
\hline
$s_9$ & 3 & $\mathbb{Z}^3$ \\
\bottomrule
\end{tabular}
\vspace*{\fill}
\end{minipage}
\caption{Example Timestamp Space Derivation (Symbols are explained in Example~\ref{exp:local-example}).}
\label{fig:algo-example}
\Description{Side-by-side layout showing an example affine program with three arrays A, B, C on the left, and the corresponding timestamp space derivation results for each statement on the right using compact interval notation.}
\end{figure}

\end{example}

After constructing the timestamp space, we attach an \emph{affine access map} that connects each timestamp to the data item it accesses.  
Let $\mathcal{R}^\ast_p$ be the largest rank among all arrays in program $p$.  
The access map $\mathcal{A}_p : \mathcal{T} \to \mathbb{Z}^{\mathcal{R}^\ast_p + 1}$
associates each timestamp $t \in \mathcal{T}$ with a data vector $m \in \mathbb{Z}^{\mathcal{R}^\ast_p + 1}$, where the first component encodes the array identifier and the following $r$ components represent its indices.  

An additional affine transformation $\mathcal{F}$ can be applied to these indices to model layout transformations or block-level reuse.  
For example, if array elements are grouped into blocks of size $\mathcal{B}$ along the last dimension, the following transform adjusts the access granularity:
\[
\mathcal{F}_{\mathcal{B}}([x_0, \ldots, x_{r - 1}, x_r]) =
[x_0, \ldots, x_{r-1}, \lfloor x_r / \mathcal{B} \rfloor].
\]

During translation, the access map shares most operations with the timestamp space, except that in access statements its range is explicitly defined:
\[
\llbracket \textbf{access}~A[e_0, \ldots, e_n] \rrbracket(\mathcal{M}; \ell; \mathcal{F})
= \{ v \mapsto \mathcal{F}([A, 0, \ldots, 0, e_0|_{\mathcal{M}}, \ldots, e_n|_{\mathcal{M}}]) : v \in \mathbb{Z}^\ell \}.
\]
This construction effectively augments the timestamp space with affine mappings that capture data access behavior.

\subsection{Counting RI Distribution as Piecewise Quasi-polynomials} \label{subsect:counting-algo}
We formulate the reuse interval distribution as a counting problem solved via Barvinok decomposition and reparametrization. We first introduce Integer Set Programming and Barvinok counting, then derive the symbolic RI representation and compute the complete distribution.

\subsubsection{Background: Integer Set Programming and Barvinok Counting}

Integer Set Programming provides the algebraic framework for manipulating affine sets and relations. We use the Integer Set Library (ISL) \cite{Verdoolaege2010isl} for affine structures and the Barvinok library \cite{barvinok} for cardinality.

ISL maps represent multi-valued functions between domain and codomain via affine transformations. When a map $R$ has identical domain and codomain $X$, it defines a binary relation $\mathcal{R}$ on $X$ where $(x, x') \in \mathcal{R} \iff x' \in R(x)$.
We use the following operations (all closed under affine integer sets and maps):
$A^{-1}$ (inverse),
$A \succ B$ (lexicographic order linking $a \in A$ to all lesser $b \in B$),
$A \succeq B, A \prec B, A \preceq B$ (similar orderings),
$\textbf{lexmax}~A$ (lexicographic maximum of each image),
$A \cap B$ (intersection),
$A \circ B$ (composition).

The Barvinok method counts integer points in parametric polytopes of the form
$\{x \in \mathbb{Q}^d \mid Ax \ge Bp + c\}$,
where $A \in \mathbb Z^{m \times d}, B \in \mathbb{Z}^{m\times n}, c \in \mathbb{Z}^m$.
It outputs \textbf{piecewise quasi-polynomials}, where each piece associates a region with a quasi-polynomial (extending ordinary polynomials with ceiling division and modulo operations) counting its integer points.

\begin{definition}[Backward Reuse Interval]\label{dfn:ri-def}
    Let $\mathcal{T}[t]$ denote the data item accessed at timestamp $t$.
    Define $t' = \max\left\{\, s \mid 0 \leq s < t,\ \mathcal{T}[s] = \mathcal{T}[t]\,\right\}$ as the most recent prior timestamp accessing the same data item.
    The backward reuse interval at time $t$ is $\ri_t := t - t' = \left|\, \left\{\, i\ \Big|\ t' < i \leq t\, \right\} \,\right|.$
    
    If no such $t'$ exists, $\ri_t := \infty$ (cold miss).
    We call the corresponding set $\left\{\, i\ \Big|\ t' < i \leq t\, \right\}$ the \textbf{ri-set} at $t$, denoted $\textbf{ri-set}(t)$, with $\ri_t = |\textbf{ri-set}(t)|$.
\end{definition}

\subsubsection{Symbolic RI Analysis}

We construct \textbf{ri-set} as an intersection of affine relations using fiber partitioning. 
A single-valued map \(A: X \to Y\) induces the \textbf{fiber relation} $K_A := \{(x, x') \in X \times X \mid A(x) = A(x')\} = A^{-1} \circ A$, which partitions $X$ into equivalence classes (fibers) of equal image.
Intersecting $K_A$ with another binary relation $R$ yields timestamps that satisfy both constraints.

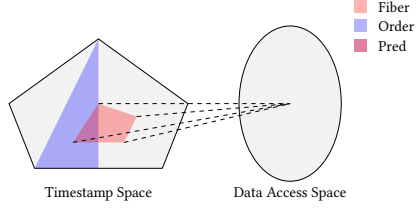
\begin{figure}[htbp]
\centering
\resizebox{0.4\linewidth}{!}{
\begin{tikzpicture}[scale=3]
  \coordinate (P1) at (0,0);
  \coordinate (P2) at (1,0);
  \coordinate (P3) at (1.2,0.5);
  \coordinate (P4) at (0.5,1);
  \coordinate (P5) at (-0.2,0.5);
  \path[fill=gray!10] (P1) -- (P2) -- (P3) -- (P4) -- (P5) -- cycle;
  \draw (P1) -- (P2) -- (P3) -- (P4) -- (P5) -- cycle;

  \draw[fill=gray!10] (2,0.5) ellipse [x radius=0.4, y radius=0.6];

  \path[fill=blue, opacity=0.3] (P1) -- (0.5,0) -- (0.5,1) -- (P1);

  \fill[red, opacity=0.3] (0.3,0.2) -- (0.7,0.2) -- (0.8,0.4) -- (0.5,0.5) -- cycle;

  \draw[dashed] (0.3,0.2) -- (2,0.5);
  \draw[dashed] (0.7,0.2) -- (2,0.5);
  \draw[dashed] (0.8,0.4) -- (2,0.5);
  \draw[dashed] (0.5,0.5) -- (2,0.5);

  \fill[red, opacity=0.3] (2.5,1.2) rectangle (2.6,1.3);
  \node[anchor=west] at (2.65,1.25) {Fiber};
  \fill[blue, opacity=0.3] (2.5,1.05) rectangle (2.6,1.15);
  \node[anchor=west] at (2.65,1.1) {Order};
  \fill[purple, opacity=0.5] (2.5,0.9) rectangle (2.6,1.0);
  \node[anchor=west] at (2.65,0.95) {Pred};

  \node at (0.5,-0.2) {Timestamp Space};
  \node at (2,-0.2) {Data Access Space};
\end{tikzpicture}
}
\caption{Diagram of the \textbf{prev} map construction.}\label{fig:prev}
\Description{A diagram showing the relationship between timestamp space and data access space for the prev map construction, illustrating how memory accesses are ordered in time.}
\end{figure}

Let $\mathcal{T}$ be the timestamp space and $\mathcal{A} : \mathcal{T} \to \mathcal{D}$ map each timestamp to its accessed data item.
The predecessor set is
$$
\textbf{prev}(\mathcal{A}, \mathcal{T}) = \left(\mathcal{A}^{-1} \circ \mathcal{A}\right)\cap (\mathcal T \succ \mathcal T),
$$
intersecting the fiber relation with lexicographic order (see \Cref{fig:prev}).
The \textbf{ri-set} is then
$$
\textbf{ri-set}(\mathcal{A}, \mathcal{T}) = \left(\textbf{lexmax}~\textbf{prev}(\mathcal{A}, \mathcal{T}) \circ (\mathcal T \prec \mathcal T)\right)\cap  (\mathcal T \succeq \mathcal T).
$$
Since \textbf{ri-set} is an affine map, we compute each RI value via Barvinok as $\ri_t = |\textbf{ri-set}(t)|$. 

\subsubsection{Re-count RI Distribution with Re-parametrization}\label{subsubsect:affine-set-counting}
Computing all RI values as quasi-polynomials is not enough; we must also count the occurrences of each value.  Enumerating all RI values is not feasible for symbolic parameters.  All hope is not lost -- we can exploit the \emph{piecewise structure} of the domain to count efficiently,  aggregating occurrences within each affine region instead of iterating over individual timestamps.

\Cref{algo:reparam} takes piecewise quasi-polynomial input $\{(S_i, \mathit{QP}_i)\}$, where $S_i$ is a subset of timestamps and $\mathit{QP}_i$ is its RI formula.
Since disjoint subsets may share RI values, we group them by value and sum their cardinalities.
The algorithm operates in two phases: 
(1) \emph{Reparameterization} (lines 4--13): For constraints $Ax \ge Bp + c$ with dimension vector $x$ and parameter vector $p$, we move any dimension $x_i$ appearing in $\mathit{QP}_i$ from $x$ to $p$, resulting in $x'$ and $p'$. The transformed constraints $A' x' \ge B' p' + c$ preserve the polytope structure, normalizing quasi-polynomials to group subsets by RI value in map $T$.
(2) \emph{Counting} (lines 14--17): We compute cardinality $c_q$ for each grouped subset to obtain occurrence counts, outputting distribution pairs $(S_j, c_j)$ with symbolic counts $c_j$.

\begin{algorithm}[htbp]
\small
\setlength{\abovedisplayskip}{3pt}
\setlength{\belowdisplayskip}{3pt}
\begin{algorithmic}[1]
\Require Piecewise quasi-polynomial $\{(S_i, \mathit{QP}_i)\}_{i=1}^n$ where $S_i$ is a domain subset and $\mathit{QP}_i$ is a quasi-polynomial
\Ensure Distribution $\mathit{Dist}$ as $\{(S_j, c_j)\}$ where $c_j$ is the cardinality formula
\State Initialize map $T$ with default value $\emptyset$ for each entry
\For{each $i \in \{1, \ldots, n\}$}
    \State Let $D_i, P_i$ be the dimension and parameter domains of $S_i$
    \State $D' \gets \emptyset$, $P' \gets P_i$
    \For{each $d \in D_i$}
        \If{$d \in \mathrm{var}(\mathit{QP}_i)$}
            \State $P' \gets P' \cup \{d\}$ \Comment{Move dimension to parameters}
        \Else
            \State $D' \gets D' \cup \{d\}$
        \EndIf
    \EndFor
    \State Let $S'_i$ be the reprojection of $S_i$ onto $(D', P')$
    \State Let $\mathit{QP}'_i$ be the normalized $\mathit{QP}_i$ on $S'_i$
    \State $T[\mathit{QP}'_i] \gets T[\mathit{QP}'_i] \cup S'_i$ \Comment{Group subsets by RI value}
\EndFor
\State $\mathit{Dist} \gets \emptyset$
\For{each $q \in \mathrm{keys}(T)$}
    \State Compute $c_q \gets \text{cardinality}(T[q])$ \Comment{Count occurrences of RI value $q$}
    \State $\mathit{Dist} \gets \mathit{Dist} \cup \{(T[q], c_q)\}$
\EndFor
\State \Return $\mathit{Dist}$
\end{algorithmic}
\caption{Re-parameterization for computing RI distribution}
\label{algo:reparam}
\end{algorithm}

\subsubsection{Infinite Repeat}

To model infinite repeat, we add a repetition dimension $r \in [0, R)$ where $R$ is a symbolic parameter representing the number of iterations. This makes both the timestamp space and RI counts grow linearly with $R$. To compute the limiting RI portions as $R \to \infty$, we apply L'Hôpital's rule: since both the numerator (RI count) and denominator (total access count) are polynomials in $R$, the limit equals the ratio of their leading coefficients. This allows efficient symbolic computation of infinite-repeat RI portions without explicit limit evaluation.

\section{Hardness of Reuse Interval Distribution Derivation}

Having established our algorithmic framework, we now examine the inherent computational complexity of RI distribution derivation before evaluating our implementation. This analysis motivates why even efficient tools face fundamental limits on certain program classes.

We demonstrate that deriving RI distributions for affine programs is computationally hard. While one might expect hardness only for programs with complex control flows, we show that the problem remains intractable even for surprisingly simple structures: \emph{branch-free\footnote{By ``branch-free'', we mean there are no if--else statements.} nested loops with polynomial-bounded memory}.

Our reductions exploit a fundamental property: \emph{accessing a memory location intercepts the original trace, splitting long reuse intervals into shorter ones}. This interception mechanism allows us to encode computational problems into RI distributions: a satisfying assignment induces strategic memory accesses that produce short RIs, while unsatisfiable formulas yield only long intervals.

We establish hardness by reducing 3-SAT to RI queries. For a 3-SAT formula
\[
\phi = C_1 \land C_2 \land \cdots \land C_m,
\]
where $C_i = l_{i1} \lor l_{i2} \lor l_{i3}$ and each literal is over variables $x_k$ ($k \in [n]$), we encode each clause by
\[
\left\lceil \frac{v_{i1} + v_{i2} + v_{i3}}{3} \right\rceil,
\qquad
v_{ij} := 
\begin{cases}
x_k & l_{ij}=x_k,\\
1 - x_k & l_{ij}=\lnot x_k.
\end{cases}
\]
The value is $1$ when the clause is satisfied and $0$ otherwise. Since $m \le O(n^3)$, we require only polynomial memory.

\begin{figure}[t]
    \centering
    \begin{subfigure}[t]{0.47\linewidth}
        \centering
        \begin{lstlisting}[basicstyle=\footnotesize\ttfamily]
access(A[m]);
for (int i1 = 0; i1 <= 1; i1++)
  for (int i2 = 0; i2 <= 1; i2++)
    ...
      for (int in = 0; in <= 1; in++)
        access(A[C1 + C2 + ... + Cm]);
access(A[m]);
        \end{lstlisting}
        \caption{Program used in the NP-completeness reduction.}
        \label{fig:np-reduction-prog}
    \end{subfigure}
    \hfill
    \begin{subfigure}[t]{0.47\linewidth}
        \centering
        \begin{lstlisting}[basicstyle=\footnotesize\ttfamily]
for (int i1 = 0; i1 <= 1; i1++)
  for (int i2 = 0; i2 <= 1; i2++)
    ...
      for (int in = 0; in <= 1; in++) {
        access(A[m]);
        access(A[C1 + C2 + ... + Cm]);
        access(A[m + 1]);
      }
        \end{lstlisting}
        \caption{Program used in the \#P-completeness reduction.}
        \label{fig:sp-reduction-prog}
    \end{subfigure}
    \caption{Side-by-side programs for the hardness reductions.}
    \label{fig:two-reduction-programs}
\end{figure}

\begin{theorem}
\label{thm:affine-np}
Determining whether a specific RI value appears in a branch-free affine program with polynomial-bounded memory is NP-complete.
\end{theorem}

\begin{proof}
We reduce 3-SAT to RI detection using the program in \Cref{fig:np-reduction-prog}.  
The nested loops enumerate all $2^n$ truth assignments.  
If $\phi$ is satisfiable, at least one iteration evaluates $\sum C_i = m$, so the program accesses $A[m]$ inside the loop. This \emph{intercepts} the trace, producing an RI value $< 2^n+1$.  
If unsatisfiable, every iteration accesses a distinct location, so the only RI between the boundary accesses is exactly $2^n+1$.  
Thus RI-detection is NP-complete.
\end{proof}

\begin{theorem}
\label{thm:affine-sp}
Counting the number of occurrences of a specific RI value in a branch-free affine program is \#P-complete.
\end{theorem}

\begin{proof}
We reduce \#3-SAT to RI counting using the program in \Cref{fig:sp-reduction-prog}.  
For each satisfying assignment ($\sum C_i = m$), the two consecutive accesses to $A[m]$ produce exactly one RI of value~$1$.  
For unsatisfying assignments, all three accesses touch different addresses, producing no RI of value~$1$.  
Thus the number of RIs of value~$1$ equals the number of satisfying assignments, proving \#P-completeness.
\end{proof}

These results indicate that deriving the full RI distribution is fundamentally difficult—even in branch-free affine programs with polynomial memory. The trace-interception mechanism has enough expressive power to encode NP- and \#P-complete problems, implying that the hardness arises from the expressive loop-and-access structure itself rather than from control-flow complexity or unbounded memory.

As shown in \Cref{sect:evalution}, although the general problem is NP-hard, our compiler achieves practical performance because (1) affine structure yields low-dimensional polytopes, (2) Barvinok’s algorithm performs well in such settings, and (3) most benchmarks exhibit few distinct RI values. The hardness primarily affects extreme cases with high loop dimensionality or adversarially constructed access patterns.

\section{Evaluation}
\label{sect:evalution}

This section evaluates the time cost of the new theory and its compiler implementation, as well as the accuracy in approximating the miss ratio of LRU caches.

\subsection{Experimental Setup}
\label{sec:padding}

We evaluate our approach using two benchmark suites: a collection of 11 Einsum loop kernels \cite{blacher2024einsum} and 30 programs from the Polybench suite \cite{polybench}.  The Einsum loops are commonly used in machine learning kernels and quantum circuit (tensor-network) simulation~\cite{flash,HPC-workload,krause2024hpbp,gray2024hyperopt,ibrahim2022constructing,64hd-q4z5}.  We omit 6 out of 17 example kernels of the Einsum benchmark because their data reuse pattern is trivial to analyze, e.g., cache-block reuse only. Instead of using a full loop, we split the Einsum loops and cache intermediate results that are repeatedly used, as suggested in \cite{gray2024hyperopt}.  An input program is compiled and analyzed as described in \Cref{sec:compiler}.
In addition to the original code, we apply first loop fusion transformations and then the analysis.  We use \texttt{mlir-opt} to apply the \texttt{affine-loop-fusion} pass \cite{mlir_passes}. It works specifically on loops expressed in the MLIR affine dialect.  Due to a known bug in the fusion cost model in previous versions, we require LLVM version 21 or higher.


We use cache simulations by Cachegrind \cite{valgrind}. Execution time data is collected from a 64-core AMD EPYC 7773X processor. Hardware counter experiments are run on a 20-core Nvidia GB10\footnote{We bind test programs to GB10's most capable Cortex-X925 core.}. All programs are compiled with the \texttt{-O3} flag to reduce stack spilling and other removable memory accesses.

We apply array padding manually in all programs to avoid conflict misses. The innermost (contiguous) dimension is padded to $8p$, where $8$ is the cache block size in number of data elements, and $p$ is the smallest prime number such that $8p$ is greater or equal than the original array dimension. The satisfies the alignment that the last dimension is a multiple of block size. For all outer dimensions except for the outermost, we round up to the nearest prime number greater than or equal to the dimension size.


\subsection{Miss Count and Miss Ratio Accuracy}
\label{sect:miss-count}

We first compare our prediction with Cachegrind's simulation.  
\Cref{fig:einsum-misscount} presents a comparison across cache sizes by miss ratio curves. For Einsum loops, the predicted miss ratio curve matches that of simulated fully associative cache, especially for cache sizes over a hundred bytes, for 10 of the 11 tests.  The exception is \texttt{weighted\_model\_counting\_pattern}, which shows a moving-cliff error, where it correctly predicts miss ratios but underestimates the cache size at which the drop occurs (see \Cref{sec:LRU_approx}).  


\begin{figure}
    \centering
    \includegraphics[width=1.0\textwidth]{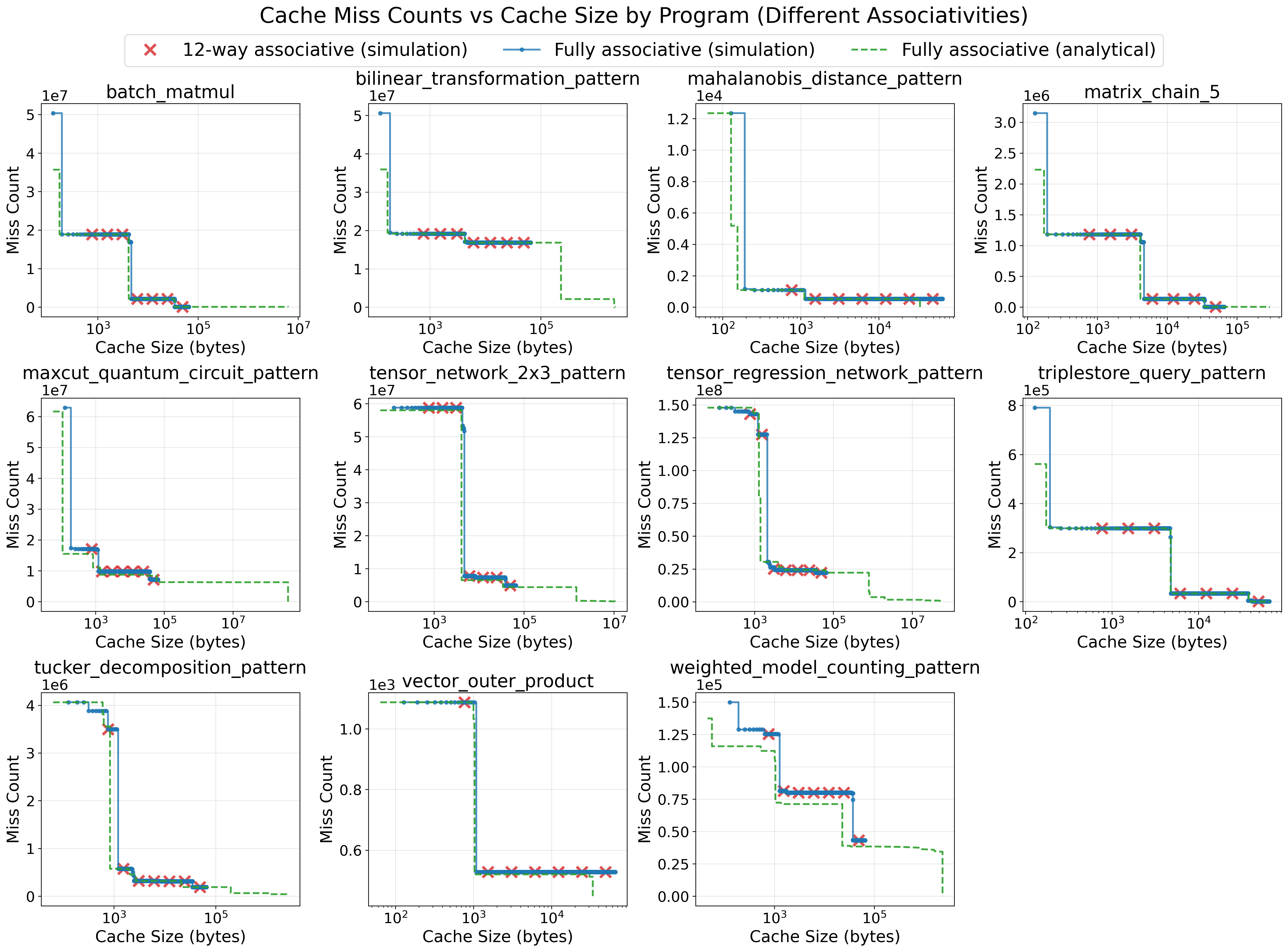}
    \caption{Tensor Contraction Kernel Miss Count Curves}
    \Description{A graph showing the predicted and actual miss ratio curve various tensor operations, using two different cache configurations.}
    \label{fig:einsum-misscount}
\end{figure}

Due to space limits, we use three cache sizes, 6KB, 12KB, and 48KB, for the 30 Polybench tests and show the results together in one table in \Cref{fig:prediction_heatmap}.  For each cache size, we compare fully assocativity and 12-way set associativity.  For each cache configuration, the table shows two miss ratios: the simulation and the prediction.   The last four columns show the comparison of the averages across three cache size.

\begin{figure}[htbp]
  \centering
  \includegraphics[width=\linewidth]{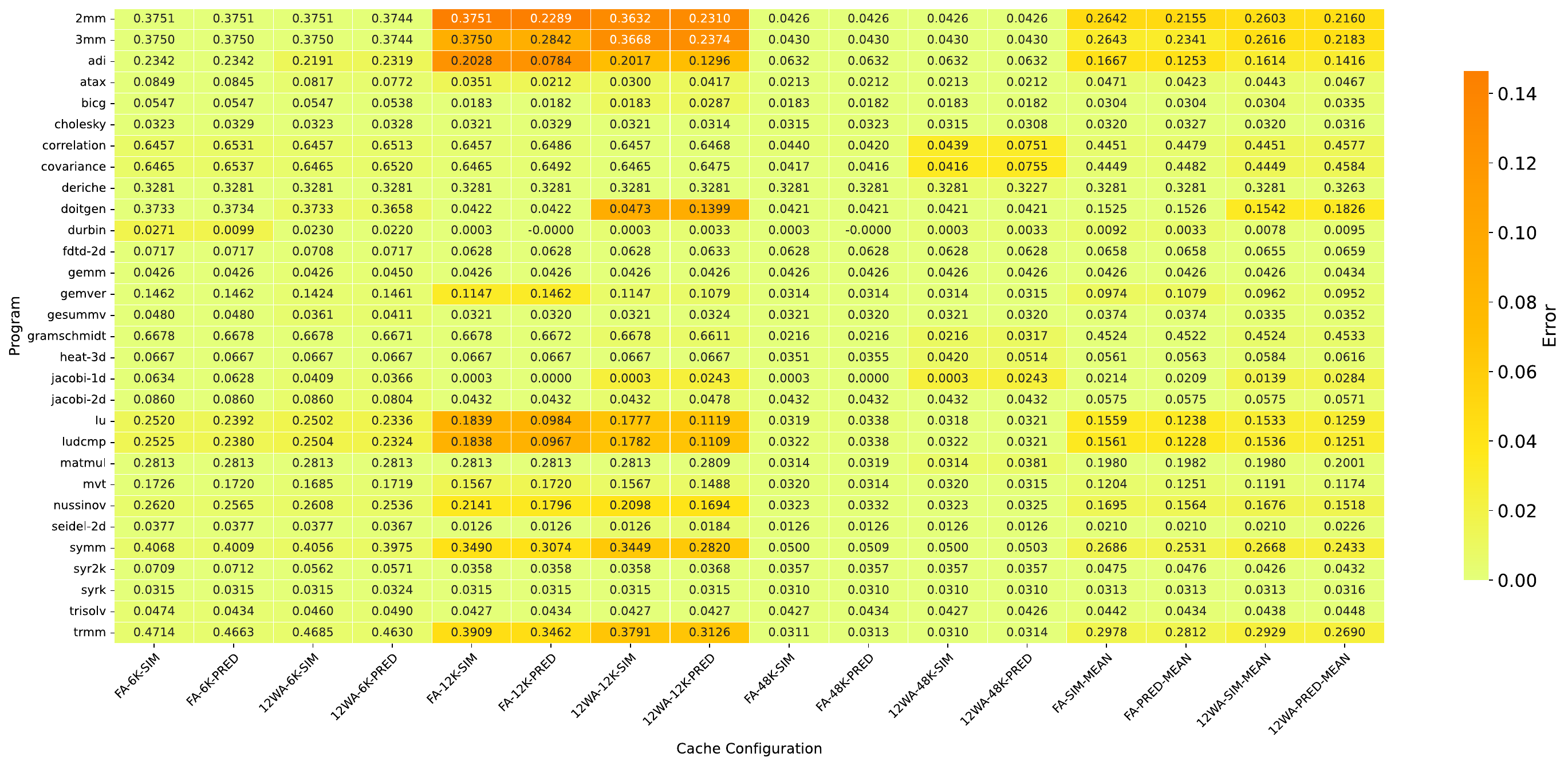}
  \caption{Prediction error across 30 Polybench benchmarks under varying cache configurations. Columns represent fully associative (FA) and 12-way associative (12WA) caches at different sizes. The cell color reflects the amount of prediction error, while the digits in each cell indicate the simulated miss ratio value.}
  \Description{A colored table showing the relative prediction error across 30 polybench programs and various cache configurations. The x-axis shows eight cache configurations alternating between fully associative (FA) and 12-way associative (12WA) caches at four different cache sizes (3KB, 6KB, 12KB, 48KB), plus mean and median statistics. The y-axis lists all 30 polybench programs alphabetically from 2mm to trmm. Each cell displays the absolute error, calculated as the absolute difference between predicted and simulated miss counts divided by total memory accesses. The color intensity (using the Wistia colormap from yellow to purple) indicates the magnitude of the error, with lighter yellow representing lower errors (better predictions) and darker purple representing higher errors. The colored table reveals that prediction accuracy varies across programs and cache configurations, with most programs showing low errors (light yellow cells) indicating accurate predictions, while certain programs and configurations exhibit higher errors (darker cells).}
  \label{fig:prediction_heatmap}
\end{figure}

The table is color coded by the error in prediction, which is the absolute difference between the number of predicted and simulated misses over the total access count.  Across all 180 comparisons,  the maximal error is 14\%, which is the top value of the color bar on the right.  The majority, 60\%, is less than 1\%, and most of these, 45\%, is less than 0.1\%.  Only 13 comparisons, 7.2\% of tests, have an error 5\% or greater, and they all happen at the 12KB cache size.


On average across three cache sizes, the prediction is highly accurate for full associativity, less than 0.01\% in 12 tests, between 0.01\% and 1\% in 9 tests, and between 1\% and 5\% in the remaining 9 tests.  The average error across all 30 tests is 1.1\%.  In comparison, the average miss ratio is 16\%.


\paragraph{Set-associative Cache}
Set conflicts can significantly increase the miss count. In a 12-way set-associative cache, if the stride of array accesses aligns such that the contiguous dimension of the array shares a common factor with 12, the array data may have a biased mapping that only utilize a small portion of the sets.

Our padding scheme (\Cref{sec:padding}) effectively removes conflict misses such that the set-associative cache performs the same as fully associative cache.  This is shown by the miss count curves for the 11 Einsum tests in \Cref{fig:einsum-misscount}, where the miss ratios of set-associative caches, marked by $\times$, fall on the miss ratio curve of the fully associative cache.  


Polybench is tested on three sizes of set-associative caches, shown in color-coded \Cref{fig:prediction_heatmap}, with the last two columns showing the average across three cache sizes.  Overall in average, the miss ratio is similar, 15.6\% in 12-way vs 16.1\% in fully associative cache.  The accuracy is also similar, with less than 0.01\% error in 10 tests, between 0.01\% and 1\% in 12 tests, and between 1\% and 5\% in the remaining 9 tests.  The average error is 1.3\%.


Statistical techniques have been developed to model set-associative cache. 
Smith~\cite{Smith:ICSE76} gave the classic formula for using reuse distances to estimate conflict errors in set-associative caches, which is used in later tools, including the HPCToolkit~\cite{MarinM:SIGMETRICS04}.  There are two other models: Set-RD~\cite{Nugteren+:HPCA14,SenW:SIGMETRICS13} and mapped footprint~\cite{Luo+:MEMSYS18}.  For regular loop code and the high associativity used by modern processors, we found that array padding is sufficient for it to perform similarly to a fully associative cache.



\paragraph{Loop Fusion Optimization}
Loop fusion is widely studied, for example, in ~\citep{DingK:JPDC04,SpaceTime:PLDI02}.
We test our programs after MLIR's \texttt{affine-loop-fusion}, which has two fusion strategies\footnote{Descriptions were summarized from MLIR's document (https://mlir.llvm.org/docs/Passes/\#-affine-loop-fusion)}: producer-consumer fusion and sibling fusion. Producer-consumer fusion focuses on merging loop pairs in which the producer loop writes to a memory reference that the consumer loop subsequently reads. Sibling fusion addresses loop pairs with not necessarily data dependence but those that have the same iteration. We use the default computation overhead threshold ($30\%$) for producer-consumer fusion as shipped in LLVM 21.

Over 86\% of errors in \Cref{fig:prediction_heatmap_fusion} are less than $2\%$, and all are less than 4\% except \texttt{lu} and \texttt{ludcmp}. In these two programs, loop bounds depend on induction variables. We note another case, \texttt{adi}. Before loop fusion, it has a moving-cliff error up to 12\% caused by having multiple loop nests.  The prediction is completely accurate after loop fusion.

\begin{figure}[htbp]
  \centering
  \includegraphics[width=\linewidth]{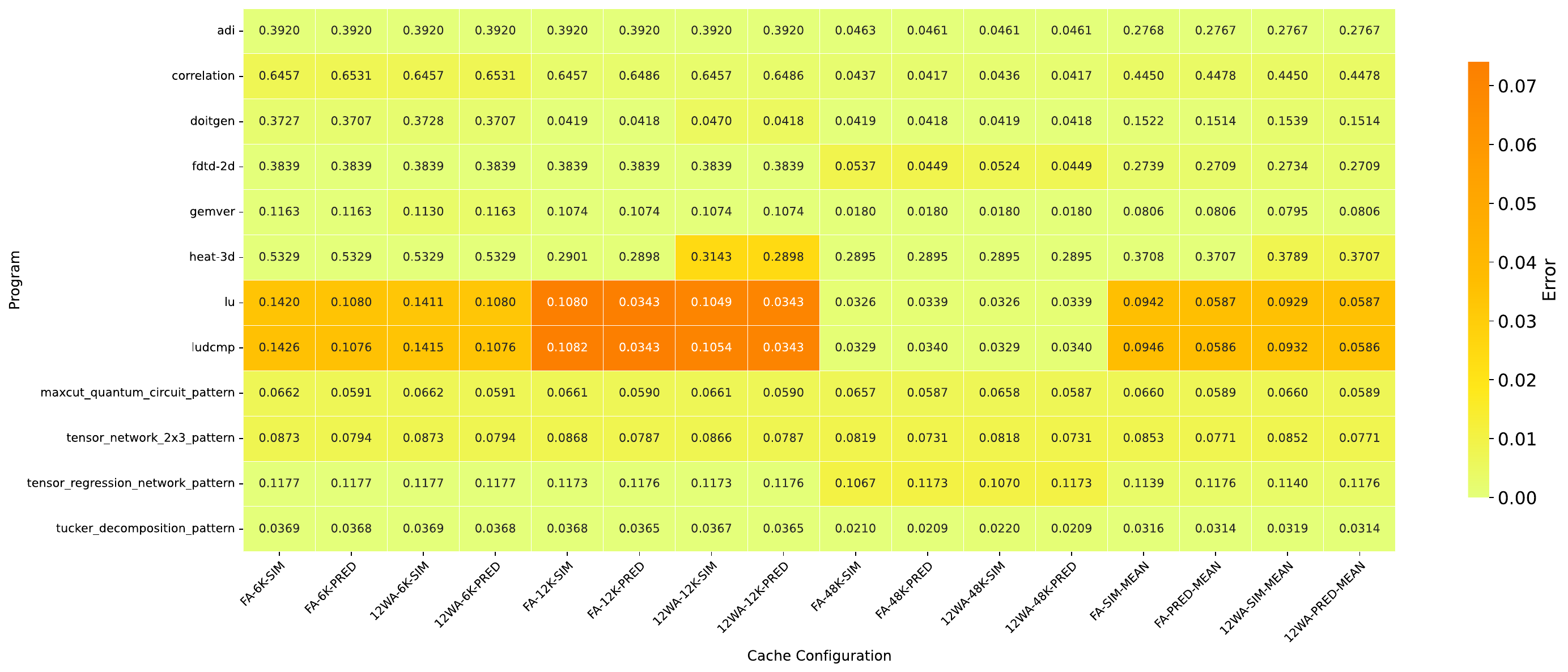}
  \caption{Prediction error across 12 programs affected by MLIR affine loop fusion, under varying cache configurations. Columns represent fully associative (FA) and 12-way associative (12WA) caches at different sizes. The cell color reflects the amount of prediction error, while the digits in each cell indicate the simulated miss ratio value.}
  \Description{A colored table showing the relative prediction error across 12 programs that were affected by the MLIR \texttt{-affine-loop-fusion} pass, under various cache configurations. The x-axis shows eight cache configurations alternating between fully associative (FA) and 12-way associative (12WA) caches at four different cache sizes (3KB, 6KB, 12KB, 48KB), plus mean and median statistics. The y-axis lists the 12 programs included in this dataset. Each cell displays the absolute error, calculated as the absolute difference between predicted and simulated miss counts divided by total memory accesses. The color intensity (using the Wistia colormap from yellow to purple) indicates the magnitude of the error, with lighter yellow representing lower errors (better predictions) and darker purple representing higher errors. The results show that for these 12 programs affected by MLIR’s \texttt{-affine-loop-fusion} transformation, prediction accuracy remains consistent across cache configurations, with most exhibiting low errors (light yellow cells).}
  \label{fig:prediction_heatmap_fusion}
\end{figure}

\paragraph{48KB 12-way Set Associative Cache}
This particular configuration, with 64-byte blocks (8 double-precision floating point numbers) is found on the performance cores of Intel Core processors from the 10th to the 14th generation, from Sunny Cove (in Ice Lake processors) through Golden Cove and Redwood Cove (in Alder Lake and Raptor Lake). In addition, AMD's Zen5 cores use the same associativity and size for its L1 data cache.  The actual cache replacement policy is not public.  Our simulation uses LRU (while \Cref{sec:hw-counters} will show actual miss counts).

The comparisons are shown in \Cref{fig:prediction_heatmap,fig:prediction_heatmap_fusion}.  Before loop fusion, the average miss ratio is 4.1\%, and the prediction error is 0.5\%.  After loop fusion (12 tests), the average miss ratio is 7\%, while the error is 0.4\%, with all but one less than 1\%.  The average error is 0.43\% cross all 41 tests of this cache configuration.

\paragraph{Data Movement Prediction Accuracy}
In standard definition, $x$\% accuracy means that $x$\% memory accesses are predicted correctly whether it is a hit or a miss.  For data movement analysis, the aggregate is required, while the order of hits and misses does not matter.  Thus, we transform a trace of accesses into a sequence of hits and misses and then permute it so all hits precede all misses. Given a predicted miss ratio, we generate a predicted sequence with hits followed by misses at that ratio.  We define \emph{data movement prediction accuracy} as the fraction of elements, including both hits and misses, that match between the permuted sequence and the predicted sequence.  When the miss ratio error is 0.4\%, the predictions align with 99.6\% of the permuted sequence.  The accuracy of data movement prediction is therefore 99.6\%.

In our tests, the miss ratio may differ by orders of magnitude.  The average accuracy, when defined on the miss ratio or the hit ratio, can be made unfair by one test, for example, a low relative accuracy but on a very small miss ratio.  In comparison, the data movement accuracy is normalized the same way across all tests, regardless of their hit or miss ratio.

\paragraph{The Real Effect of Imaginary Reuses}
Imaginary reuses are crucial for accuracy.  In an early, incomplete evaluation including 23 PolyBench kernels tested against simulation on 32KB fully associative cache, 
without Imaginary Reuses, the prediction error is up to 19.88\% in the worst case and 2.15\% on average. With Imaginary Reuses, the maximum error is reduced to 1.53\%, and the average error drops to 0.18\% --- a more than tenfold improvement in accuracy.


\subsection{Analysis and Prediction Costs}
\label{sec:speed}

\Cref{fig:speed} shows the time takes an EPYC 7773X to derive cache polynomials and to use them in prediction.  It labels the former as \emph{construction} and latter as \emph{prediction}.

\begin{figure}[htbp]
  \centering 
  \includegraphics[width=\linewidth]{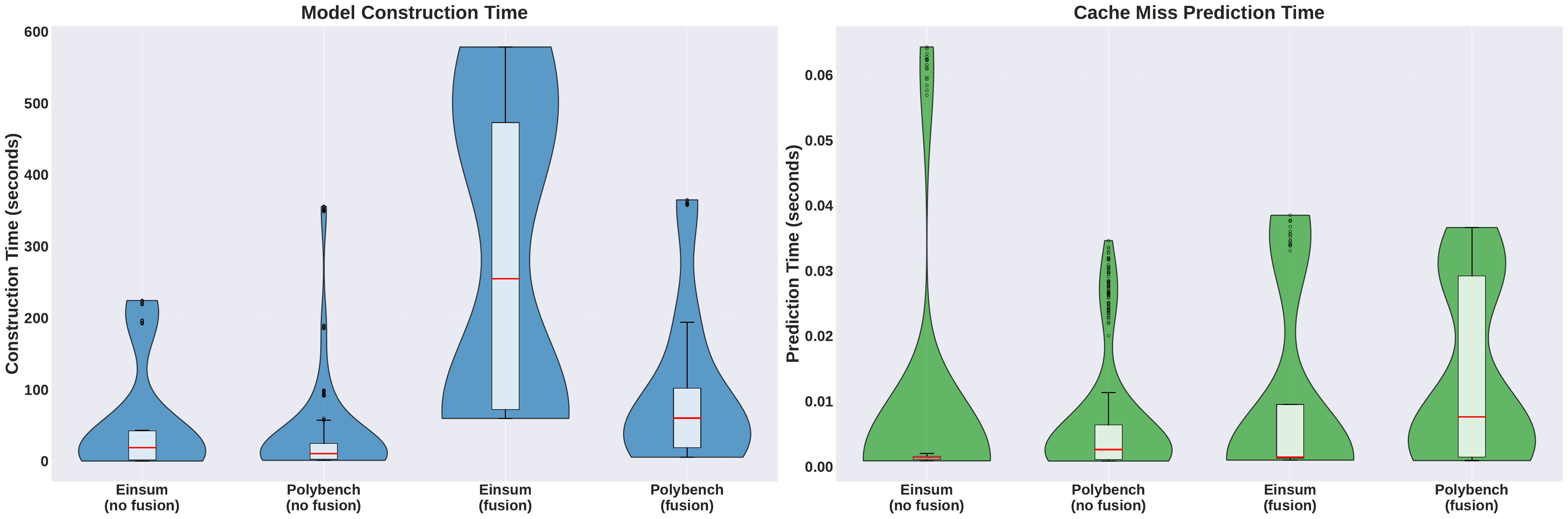}
  \caption{Violin plots of timing distributions for the model construction and prediction phases. 
  The scales differ significantly between the two phases, hence they are plotted separately for clarity.}
  \label{fig:speed}
  \Description{Violin plots comparing timing distributions between construction and prediction phases, showing that construction is much slower while prediction is nearly instantaneous.}
\end{figure}

For derivation, most programs without fusion run under $50$ seconds, with its median around $65$ seconds.  Long derivation times occur in \texttt{maxcut\_quantum\_network\_pattern}, with $12$ nesting depth, and \texttt{tensor\_regression\_network\_pattern} with $11$-nested loops, due to the high dimensionality of their timestamp spaces. Programs with loop fusion generally show an increase in \texttt{construction} time, as producer-consumer fusion may make the program more complex. 

Once the polynomials are derived, prediction is instantaneous —- all programs complete within 65 milliseconds. The \textit{prediction} time is independent of the complexity of the program itself but correlates with the number of symbolic terms.  We observe that loop fusion has minimal effects on \textit{prediction} time.

\subsection{Comparison with Hardware Counters}
\label{sec:hw-counters}

We validate our prediction methodology by comparing simulation results against hardware performance counters on Nvidia GB10's Cortex-X925. We collect measurements using ARM PMU events \texttt{l1d\_cache\_refill} and \texttt{l1d\_cache}. The Cortex-X925 employs a 64KiB 4-way set-associative L1d cache. To reduce measurement noise inherent in hardware counters for small workloads, we limit our analysis to test cases with refill counts exceeding 600,000. All experiments use the padding configuration described in \Cref{sec:padding}.

\Cref{fig:hwct} demonstrates strong agreement between simulation and hardware counter measurements. While \texttt{fdtd-2d}, \texttt{correlation}, and \texttt{covariance} exhibit slightly higher relative differences, the absolute discrepancies remain small.  In \texttt{heat-3d}, prediction is lower simulation because of the moving-cliff error (\Cref{sec:LRU_approx}), but the simulation matches with the hardware counter result.  In all other cases, our predictions are accurate compared to both the simulation and the hardware measurements.

\begin{figure}[htbp]
  \centering
  \includegraphics[width=\linewidth]{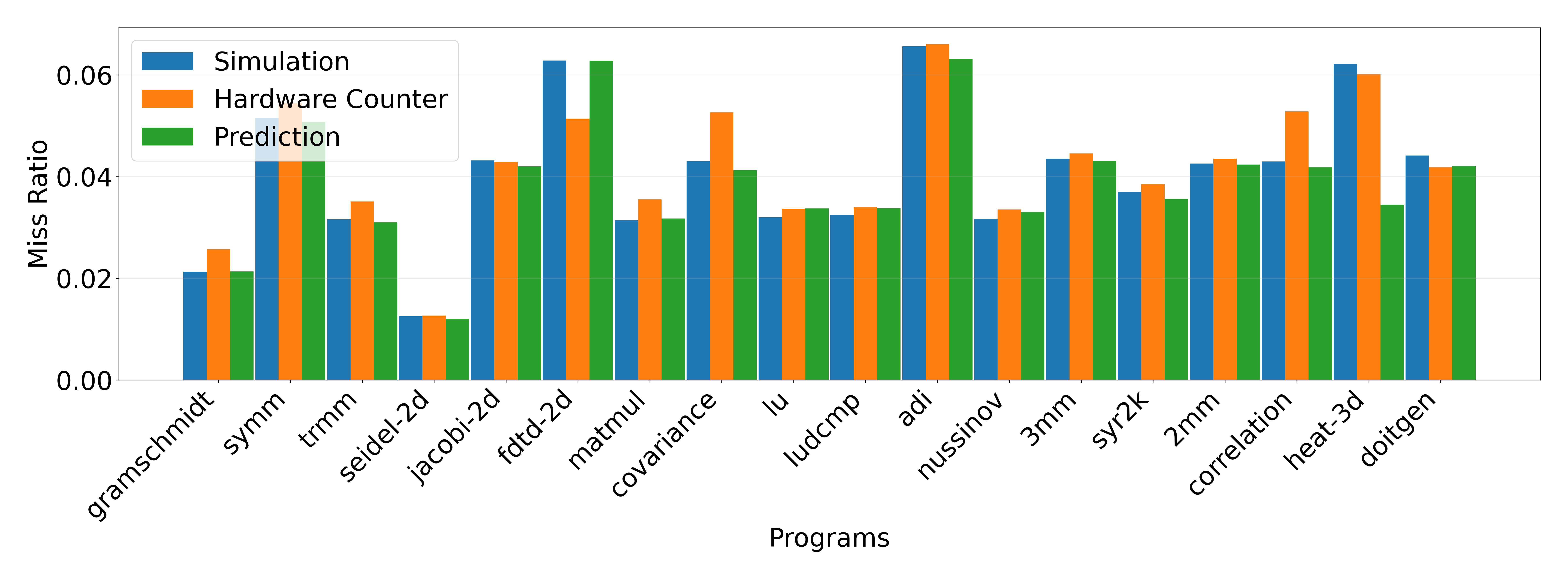}
  \caption{Comparing hardware-counter result with prediction and simulation}
 \label{fig:hwct}
 \Description{A grouped bar chart comparing miss ratios for several benchmark programs
across three measurement methods: simulation, hardware performance counters,
and symbolic prediction. Each program has three bars clustered together.
Across all programs, the predicted miss ratios (green bars) closely follow the
simulation results (blue bars), usually falling between the simulation and
hardware-counter values. The hardware-counter measurements (orange bars) tend
to be slightly higher than simulation for many benchmarks. The overall trend
shows that symbolic prediction tracks simulation with small error across all
benchmarks, demonstrating consistent alignment between predicted and empirical
miss ratios.}
\end{figure}

\subsection{Cache Min-Max Scaling}
\label{sec:min-max}

In analyzing cache demand, a well-known rule-of-thumb is the \emph{$\sqrt{2}$ rule}, which states that \emph{if you double the problem size, you need to multiply the cache size by $\sqrt{2}$ to maintain the same cache hit ratio}\citep{Hartstein+:JILP08}.  This rule has been empirically validated across workloads in database systems (scaling buffer pool sizes), file system caches, and Web caching. 

Cache polynomials can be used to derive the cache scaling like the $\sqrt{2}$ rule.  We define the \emph{cache min-max scaling} as: the minimal cache size required to bound the maximum miss ratio below a given threshold.  For a program, its min-max scaling gives the smallest cache size that guarantees the miss ratio never exceeds a specified upper bound.

Table~\ref{tbl:scaling} shows the scaling for naive matrix multiplication.  The two data columns show the maximal miss ratio on the left and the minimal cache size on the right.  They are copied from Table~\ref{tbl:ana-matmul}, where the cache size and miss ratio polynomials are directly copied here to give the min-max scaling.  Since the cache size must be integral, we take the ceiling of the constant terms when copied from Table~\ref{tbl:ana-matmul}.


\begin{table}[!htbp]
\renewcommand{\arraystretch}{1.2}
\caption{Cache Min-Max Scaling of Naive Matrix Multiplication}
\label{tbl:scaling}
\centering
\begin{tabular}{c|cc}
\toprule
Index & min cache size & max miss ratio \\ \hline
0 & $0$ & $1$ \\
1 & $1$ & $\frac{3}{4}$ \\
2 & $3$ & $\frac{1}{2} + \frac{1}{32n}$ \\
3 & $4$ & $\frac{9}{32} + \frac{1}{32n}$ \\
4 & $\frac{9n}{8} - 9$ & $\frac{1}{4} + \frac{1}{16n}$ \\
5 & $\frac{9n}{8} + 2$ & $\frac{1}{32} + \frac{1}{16n}$ \\
6 & $\frac{n^2}{8} + \frac{3n}{8} - 2$ & $\frac{3}{32n}$ \\
\bottomrule
\end{tabular}
\end{table}


The top row shows the starting point: no cache, no hits.  
With just 1 block, we bound the max miss ratio to 0.75. As we increase the cache size to 3 and then 4 blocks, the miss ratio drops to at most $\frac{1}{2}+\frac{1}{32} \approx 0.53$ and then to 10/32 (~0.31). When the required cache size scales linearly with $n$, the miss ratio is bounded as low as 3/32 (~0.09). This is the $\sqrt{2}$ rule.  Finally, the minimal cache size is precisely $\frac{n^2}{8} + \frac{3n}{8} - 2$ to capture all data reuses, that is, all cache misses are cold-start misses.

Row 4 and 5 both scale by the $\sqrt{2}$ rule.  To maintain the maximal miss ratio, the cache size increases linearly with $n$, while the data size is quadratic in $n$.  Each time the data size doubles, the minimal cache size increases by $\sqrt{2}$.

Compared to $\sqrt{2}$ rule, the min-max scaling is fully precise in the exact cache size and miss ratio.  The precision shows important differences.  First, the miss ratio may not be unqiue.  For example, Row 4 and 5 both scale by the same $\sqrt{2}$ rule, but the miss ratio differs by a factor of two.  Second, the miss ratio may not be constant.  At Row 5, the miss ratio is $\frac{1}{32}+\frac{1}{16n}$, which decreases by a factor of three from $\frac{3}{32}$ to $\frac{1}{31}$ as $n$ increases from 1 to infinity.  

\section{Related Work}
\label{sec:rel}


\paragraph{Trace-based Locality Characterization}
The iterative method to compute locality comes originally from the working-set theory.  It targets unending stochastic processes.  \citet{DenningS:CACM72} made three assumptions: (1) the sequence is infinite, (2) the underlying stochastic mechanism is stationary, and (3) the events become uncorrelated as they are far apart.  The working-set theory computes the limit of the working-set size.  For example, a \emph{nonrecurrent} page is referenced at most a finite number of times.  They make no contribution to the limit working-set size.  As we refer to as the cold-miss dilemma (\Cref{sec:inf-repeat}), this assumption does not hold in program analysis where there is a nonzero portion of infinite RIs.

For a finite-length execution trace, the Higher Order Theory of Locality (HOTL) \citep{Xiang+:ASPLOS13} defined the footprint as the average working-set size of all windows of the same length (which can be computed in linear time for all window lengths).  The effect of first-touch accesses is included in the footprint as computed by the Xiang formula~\citep{Xiang+:PACT11}, which uses first- and last-access times but not infinite RIs.

The fastest techniques in practice are based on RIs, because RIs can be efficiently sampled~\citep{BeylsD:ICCS04,Wang+:CCGrid15,Xiang+:ASPLOS13}.
\citet[\S 3.3]{Yuan+:TACO19} showed that these models are mathematically equivalent to Denning Recursion.  These include the time-to-RD conversion by the Shen-Shaw formula~\citep{Shen+:POPL07} and its uses in the programming tool SLO~\citep{BeylsD:HPCC06} and Featherlight RD using hardware counters~\citep{Wang+:HPCA19}, StatStack~\citep{EklovH:ISPASS10}, and Average Eviction Time (AET) based analysis of storage caches~\citep{Hu+:TOS18,Pan+:TOS21,Xiao+:ICS23}, and multi-level caches~\citep{Ye+:TACO17}.  While the Shen-Shaw formula and StatStack are based on statistical reasoning, and AET is based on a kinetic model of cache eviction, they are all fundamentally variations of Denning Recursion.  


\paragraph{Problem Size Scaling}
By profiling results for a number of program inputs, a common approach is to extrapolate and predict the effect of other program inputs.  Such techniques have been developed for both program analysis~\citep{Zhong+:TOPLAS09,Fang+:PACT05,MarinM:SIGMETRICS04} and architecture or cache scaling analysis~\citep{Wu+:ISCA13}. Profiling reuse distance incurs a high runtime cost, which can be reduced by approximation~\citep{Olken:LBL81,CascavalP:ICS03,Zhong+:TOPLAS09,Wires+:OSDI14}, sampling~\citep{ZhongC:ISMM08,Schuff+:PACT10} and parallelization~\citep{Schuff+:PACT10,Cui+:IPDPS12,Niu+:IPDPS12}.  These techniques model LRU directly and do not have moving-cliff errors. Being data driven, they are not limited to loop code.  


\paragraph{Affine Loop Analysis}
\citet{805152} presented an analytical approach to modeling for both set-associative and fully-associative caches. They targeted perfectly nested loops with constant loop bounds and affine array references. Their approach represents each array reference's physical memory location (with the addition of offset between the arrays) as an affine expression and uses modulo arithmetic to determine how memory accesses map into cache sets, characterizing interference through an overlap metric that counts cache lines competing for each set. Their approach achieves an accurate prediction, typically within $15\%$ of the simulation result. In comparison, our approach provides comparable accuracy while handling symbolic loop bounds, imperfectly nested loops, branches, and variable cache block sizes. 

Prior techniques \cite{Kelly+:LCPC99,BeylsD:JSA05,GMM:TOPLAS99,Bao+:POPL18,Falcon:PLDI24, 10.1007/s00453-006-1231-0} have shown that the reuse distance (LRU stack distance) computation problem for certain affine programs can be converted into a counting problem involving integer points in the union or intersection of parametric polytopes.  The range of consecutive accesses can be modeled as an affine constrained set, and reuse distance values can be computed as the cardinality of this set. 
This work also uses integer set for affine loops, but the purpose is reuse intervals rather than cache behavior.  We have shown that the problem is NP-hard, so there is no polynomial-time solution unless P=NP.  Our solution takes a similar approach as previous techniques to identify data reuse.  For a data reuse, it is simple to obtain its RI as the length of the reuse window.

Previous work on PolyCache \cite{Bao+:POPL18} reports runtime performance across Polybench kernels.
On their system, the mean time is 76.2 seconds, median 37.0 seconds, minimum 1 second (\texttt{jacobi-1D}, \texttt{seidel}), and maximum 323 seconds (\texttt{symm}).  In comparison, our compiler takes 65 seconds on average for derivation and less than 0.1 second for prediction.  
PolyCache requires numerical loop bounds, cache size, and cache block size.  It does not have a symbolic derivation phase as our method does.

Finally, previous techniques do not express locality in polynomial form. The min-max scaling in \Cref{tbl:scaling} is the first demonstration of automatically deriving cache-performance scaling expressed through quadratic and reciprocal functions.  In contrast to traditional approaches that rely on empirical rules such as the $\sqrt{2}$ rule for capacity scaling~\citep{Hartstein+:JILP08}, min-max scaling captures the exact miss ratio, which can be inherently non-constant.


\section{Summary}

This paper introduces algebraic locality, a theory that derives cache size and miss polynomials from symbolic RI distributions.  It uses Infinite Repeat and converts first-touch accesses into imaginary reuses, yielding a complete RI distribution. Applying Denning Recursion then produces cache-size and miss-ratio polynomials in closed-form expressions that capture locality as algebraic functions of both program and cache parameters. In addition, 
we build a compiler for the affine MLIR dialect that converts loop nests into parametric polytopes and computes symbolic RI distributions using integer-set methods and Barvinok decomposition. 

On 41 scientific kernels and tensor operations from PolyBench and Einsum suites, deriving cache polynomials takes 41 seconds on average for all test before loop fusion and 224 seconds for 12 fused kernels.  After derivation, predicting cache misses for any program input and cache configuration takes under 1 ms. The approach achieves an average accuracy of 99.6\% for the typical L1D set-associative cache.  In addition, the paper shows the first compiler analysis for cache-performance scaling expressed symbolically in quadratic and reciprocal expressions.  The min-max scaling for affine loops is fully precise, unlike traditional scaling rules such as the $\sqrt{2}$ rule, well known for database systems (scaling buffer pool sizes), file system caches, and Web caching.

\begin{acks}

The authors wish to thank Jack Cashman for early participation in the project and Sree Pai 
for help with the presentation of the paper.



\end{acks}

\bibliographystyle{ACM-Reference-Format}
\bibliography{Bib/yifan, Bib/extra, Bib/all}

\end{document}